%% file: main.tex
\documentclass[colorlinks=true]{llncs}
\usepackage[T1]{fontenc}
\usepackage{graphicx}
\usepackage{orcidlink}
\newcommand{\macrospath}{./macros}

\input{\macrospath/macros-preamble}
\setboolean{lncsstyle}{true}
\setboolean{withproofs}{true}
\input{\macrospath/macros}
\input{\macrospath/macros-appendix}
\input{\macrospath/macros-diagrams}

\input{\notespath/local-macros}

\usepackage{listings}

\begin{document}
\title{Positive Sharing and Abstract Machines}
%
%
\author{Beniamino Accattoli\inst{1}\orcidlink{0000-0003-4944-9944} \and
Claudio Sacerdoti Coen\inst{2}\orcidlink{0000-0002-4360-6016} \and
Jui-Hsuan Wu\inst{3}\orcidlink{0000-0001-5880-5379}}

\authorrunning{Accattoli et al.}
\institute{Inria \& LIX, École Polytechnique, UMR 7161, Palaiseau, France
\and
Alma Mater Studiorum - Università di Bologna, Italy
\and
CNRS, LIP, ENS de Lyon, France
}
%
%
%
\maketitle              
\begin{abstract}
\input{00-abstract}
\end{abstract}

\input{01-introduction}
\input{02-positive_calculus}

\input{03-natural_machine}
\input{04-preliminaries_about_machines}
\input{05-sliced-machine}
\input{06-mechanical_bisimulation}

\input{07-complexity}
\input{99-conclusions}

\subsubsection{\ackname} The second author is funded by the INdAM/GNCS project MARQ and the third by the ANR project RECIPROG (ANR-21-CE48-019).

\newpage
\bibliographystyle{splncs04}
\withoutproofs{\bibliography{\macrospath/biblio,local-bib}}
\withproofs{\bibliography{main.bbl}}

\withproofs{\newpage
\appendix
\setboolean{appendix}{true}
\input{APP-00-implementation}
\input{APP-01-proofs}
}
\end{document}

%% file: macros/macros-preamble.tex
\usepackage{ifthen}
\usepackage{xspace}

\newboolean{talk}
\setboolean{talk}{false}
\newboolean{paper}
\setboolean{paper}{false}

\newboolean{IEEEstyle}
\setboolean{IEEEstyle}{false}
\newboolean{lipicsstyle}
\setboolean{lipicsstyle}{false}
\newboolean{eptcsstyle}
\setboolean{eptcsstyle}{false}
\newboolean{entcsstyle}
\setboolean{entcsstyle}{false}
\newboolean{sigplanstyle}
\setboolean{sigplanstyle}{false}
\newboolean{easychairstyle}
\setboolean{easychairstyle}{false}
\newboolean{scrartclstyle}
\setboolean{scrartclstyle}{false}
\newboolean{lncsstyle}
\setboolean{lncsstyle}{false}

\newboolean{acmartstyle}
\setboolean{acmartstyle}{false}
\newboolean{mfpsstyle}
\setboolean{mfpsstyle}{false}

\newboolean{needstheorems}
\setboolean{needstheorems}{false}
\newboolean{withimages}
\setboolean{withimages}{false}
\newboolean{withproofs}
\setboolean{withproofs}{true}

\newboolean{french}
\setboolean{french}{false}

%% file: macros/macros.tex
\usepackage{amssymb}
\usepackage{hyperref}

\usepackage{stmaryrd}
\usepackage{mathtools}

\usepackage{hhline}
\usepackage{xcolor}
\usepackage{cleveref}

\usepackage{bussproofs}

\newcommand{\ignore}[1]{}

\newcommand{\colspace}{@{\hspace{.3cm}}}
\newcommand{\myinput}[1]{\ifthenelse{\boolean{withimages}}{\input{#1}}{}}

\newcommand{\reflemma}[1]{Lemma~\ref{l:#1}}
\newcommand{\reflemmap}[2]{Lemma~\ref{l:#1}.\ref{p:#1-#2}}

\newcommand{\refpoint}[1]{Point~\ref{p:#1}}

\newcommand{\refthm}[1]{Thm.~\ref{thm:#1}}

\newcommand{\refsect}[1]{Sect.~\ref{sect:#1}}

\newcommand{\refapp}[1]{Appendix~\ref{sect:#1} (p.~\pageref{sect:#1})}

\renewcommand{\refeq}[1]{(\ref{eq:#1})} 
\newcommand{\reffig}[1]{Fig.~\ref{fig:#1}}

\newcommand{\refdef}[1]{Def.~\ref{def:#1}}

\newcommand{\ie}{\textit{i.e.}\xspace}
\newcommand{\eg}{\textit{e.g.}\xspace}
\newcommand{\ih}{\textit{i.h.}\xspace}

\newcommand{\ES}{\text{ES}\xspace}
\newcommand{\ESs}{\text{ESs}\xspace}

\newcommand{\defeq}{\coloneqq} 
\newcommand{\grameq}{\Coloneqq} 
\newcommand{\set}[1]{\{#1\}}

\newcommand{\nat}{\mathbb{N}}
\newcommand{\size}[1]{|#1|}

\newcommand{\mulsym}{\msym} 

\renewcommand{\l}{\lambda}
\newcommand{\isub}[2]{\{#1/#2\}}
\newcommand{\replace}[2]{#1{\shortleftarrow}#2}
\renewcommand{\isub}[2]{\{\replace{#1}{#2}\}}
\newcommand{\esub}[2]{[\replace{#1}{#2}]}
\renewcommand{\esub}[2]{[#1{\shortleftarrow}#2]}

\newcommand{\fv}[1]{{\sf fv}(#1)}
\newcommand{\afv}[1]{{\sf afv}(#1)}

\newcommand{\nf}{\mathsf{nf}} 
\newcommand{\nfo}[1]{\nf_\osym(#1)} 

\newcommand{\rootRew}[1]{\mapsto_{#1}}

\newcommand{\Rew}[1]{\rightarrow_{#1}}

\newcommand{\lRew}[1]{\; \mbox{}_{#1}{\leftarrow}\ }

\newcommand{\lto}{\lRew{}}

\newcommand{\rtom}{\rootRew{\msym}}
\newcommand{\rtoe}{\rootRew{\esym}}

\newcommand{\esym}{{\mathsf e}}

\newcommand{\rsym}{{\mathsf r}}

\newcommand{\msym}{\mathsf{m}}

\newcommand{\psym}{{\mathsf p}}

\newcommand{\shufeqext}{\shufeqext} 

\newcommand{\tom}{\Rew{\msym}}

 \newcommand{\toe}{\Rew{\esym}}

\newcommand{\torm}{\Rew{\rsym\msym}}
\newcommand{\tore}{\Rew{\rsym\esym}}

\newcommand{\tm}{t}
\newcommand{\tmtwo}{u}
\newcommand{\tmthree}{r}
\newcommand{\tmfour}{q}

\newcommand{\tmfirst}{\tm_1}
\newcommand{\tmsec}{\tm_2}

\newcommand{\tmtwop}{\tmtwo'}

\newcommand{\var}{x}
\newcommand{\varp}{x'}
\newcommand{\vartwo}{y}
\newcommand{\vartwop}{y'}
\newcommand{\varthree}{z}

\newcommand{\varfour}{w}

\newcommand{\varfive}{\varp}
\newcommand{\varsix}{\vartwop}

\newcommand{\ctxholep}[1]{\langle #1\rangle}
\newcommand{\ctxhole}{\ctxholep{\cdot}}

\newcommand{\arbctxp}[1]{\arbctxp{#1}}
\newcommand{\arbctxtwop}[1]{\arbctxtwop{#1}}

\newcommand{\evctx}{\genevctx}
\newcommand{\evctxtwo}{\evctx'}
\newcommand{\evctxthree}{\evctx''}
\newcommand{\evctxp}[1]{\evctx\ctxholep{#1}}
\newcommand{\evctxtwop}[1]{\evctxtwo\ctxholep{#1}}
\newcommand{\evctxthreep}[1]{\evctxthree\ctxholep{#1}}

\newcommand{\evctxfour}{\evctx'''}

\newcommand{\tomachhole}[1]{\leadsto_{#1}}
\newcommand{\tomach}{\tomachhole{}}
\newcommand{\tomachm}{\tomachhole{\mulsym}}

\newcommand{\tomache}{\tomachhole{\esym}}

\newcommand{\tomachseaone}{\tomachhole{\seasym_1}}
\newcommand{\tomachseatwo}{\tomachhole{\seasym_2}}
\newcommand{\tomachseathree}{\tomachhole{\seasym_3}}

\newcommand{\genv}{\symfont{E}}
\newcommand{\genvtwo}{\genv'}
\newcommand{\genvthree}{\genv''}

\newcommand{\cons}{:}

\ifthenelse{\boolean{acmartstyle}
}{
  \renewcommand{\state}{s}
  }{
\newcommand{\state}{\symfont{Q}}
  }
\newcommand{\statetwo}{\state'}
\newcommand{\statethree}{\state''}
\newcommand{\statefour}{\state'''}

\newcommand{\rename}[1]{\renamenop{#1}}
\newcommand{\renamenop}[1]{#1^\alpha}

\newcommand{\decode}[1]{\llbracket #1\rrbracket}
\renewcommand{\decode}[1]{\overline{#1}}

\newcommand{\deriv}{d}
\newcommand{\derivtwo}{e}
\newcommand{\derivthree}{f}
\newcommand{\derivfour}{g}

\newcommand{\sizehole}[2]{|#2|_{#1}}

\newcommand{\sizepr}[1]{\sizehole{\psym}{#1}} 

\newcommand{\mach}{{\sf M}}

\newcommand{\osym}{{\mathsf o}}

\newcommand{\la}[1]{\lambda #1.}

\newcommand{\rctx}{R}
\newcommand{\rctxtwo}{\rctx'}

\newcommand{\rctxp}[1]{\rctx\ctxholep{#1}}

\newcommand{\decodep}[2]{\decode{#1}\ctxholep{#2}}

\newcommand{\myproof}[1]{
\ifthenelse{\boolean{omitproofs}}{\begin{IEEEproof} Proof available but omitted for readability. \end{IEEEproof}}{#1}}

\newcommand{\gregoire}{Gr{\'{e}}goire\xspace}

\newcommand{\node}{\mathsf{n}}

\newcommand{\withproofs}[1]{\ifthenelse{\boolean{withproofs}}{#1}{}}
\newcommand{\withoutproofs}[1]{\ifthenelse{\boolean{withproofs}}{}{#1}}

\crefname{proposition}{Prop.}{Props.}
\crefname{theorem}{Thm.}{Thms.}
\crefname{lemma}{Lemma}{Lemmas}
\crefname{corollary}{Cor.}{Cors.}
\crefname{section}{Sect.}{Sects.}
\Crefname{section}{Section}{Sections}

\newcommand{\doubt}[1]{}

\newcommand{\letexp}{\mathsf{let}}
\newcommand{\letin}[3]{{\sf let}\ #1=#2\ {\sf in}\ #3}

\newcounter{numberone}
\newcounter{numberoneroman}
\newcounter{numberonealph}

\newcommand{\cbv}{CbV\xspace}

\newcommand{\tostrat}{\Rew{\symfont{str}}}

\newcommand{\tomachine}{\tomachhole\mach}

\newcommand{\mytr}[1]{\underline{#1}}
\newcommand{\auxtr}[1]{\overline{#1}}

\makeatletter
\newcommand\Copy[2]{
        \marginpar{\scriptsize \ \ \hyperlink{hl-appendix-#1}{Proof p.\,{\pageref*{appendix-#1}}}}
	\immediate\write\@auxout{\unexpanded{\global\long\@namedef{mytext@#1}{#2}
  }}%
	#2%
}

\newcommand\Paste[1]{%
        \hypertarget{hl-appendix-#1}{}\label{appendix-#1}
	\renewcommand{\inappendix}[1]{}
	\ifcsname mytext@#1\endcsname
	\@nameuse{mytext@#1}%
	\else
	``??''
	\fi
	\renewcommand{\inappendix}[1]{#1}
}
\makeatother

\newcommand{\inappendix}[1]{#1}

\newcommand{\sizep}[2]{|#1|_{#2}}

\newcommand\Crumb\mytr
\newcommand\CrumbAux\auxtr

\makeatletter
\newcommand{\exder}{%
  \def\exderW[##1]{\triangleright_{##1}\ }%
  \def\exderWO{\triangleright\ }%
  \@ifnextchar[\exderW\exderWO%
  }
\makeatother

\newcommand{\dom}[1]{\mathsf{dom}(#1)}

\newcommand{\bigo}{\mathcal{O}}

\newcommand\mydots{\hbox to .6em{.\hss.}}

\newcommand{\Lpos}{$\lambda_{\symfont{pos}}$\xspace}

\newcommand{\topos}{\Rew{\possym}}
\renewcommand{\topos}{\Rew{\osym_+}}

\makeatletter
\newsavebox{\@brx}
\newcommand{\llangle}[1][]{\savebox{\@brx}{\(\m@th{#1\langle}\)}%
  \mathopen{\copy\@brx\kern-0.5\wd\@brx\usebox{\@brx}}}
\newcommand{\rrangle}[1][]{\savebox{\@brx}{\(\m@th{#1\rangle}\)}%
  \mathclose{\copy\@brx\kern-0.5\wd\@brx\usebox{\@brx}}}
\makeatother

\newcommand{\asym}{a}

\newcommand{\symfont}[1]{\mathtt{#1}}

\newcommand{\bite}{b}
\newcommand{\bitetwo}{\bite'}

\newcommand{\tor}{\Rew{\symfont{r}}}

\newcommand{\run}{r}
\newcommand{\runtwo}{\run'}
\newcommand{\runthree}{\run''}

\newcommand{\States}{\symfont{States}}

\newcommand{\compilrelsym}{\triangleleft}
\newcommand{\compilrel}[2]{#1\compilrelsym#2}
\newcommand{\prsym}{\symfont{pr}}
\newcommand{\tomachpr}{\tomachhole{\prsym}}
\newcommand{\seasym}{\symfont{sea}}
\newcommand{\tomachsea}{\tomachhole{\seasym}}

\newcommand{\orctx}{\rctx_o}
\newcommand{\orctxtwo}{\orctx'}

\newcommand{\orctxp}[1]{\orctx\ctxholep{#1}}
\newcommand{\orctxtwop}[1]{\orctxtwo\ctxholep{#1}}

\newcommand{\irctx}{\rctx_i}
\newcommand{\irctxtwo}{\irctx'}

\newcommand{\irctxtwop}[1]{\irctxtwo\ctxholep{#1}}

\newcommand{\slice}{\symfont{S}}
\newcommand{\slicetwo}{\slice'}

\newcommand{\emptylist}{\epsilon}
\newcommand{\emslice}{\epsilon}
\newcommand{\sliceentry}[2]{#1[\replace{#2}{\cdot}]}

%% file: macros/macros-appendix.tex
\newboolean{appendix}
\setboolean{appendix}{false}
\usepackage{appendix}
\ifthenelse{\boolean{mfpsstyle}}{\capturecounter{thm}}{
	\ifthenelse{\boolean{lncsstyle}}{}{\capturecounter{theorem}}}
\usepackage{marginnote}

\newcommand{\NoteProof}[1]{\ifthenelse{\boolean{withproofs}}{
	\ifthenelse{\boolean{appendix}}{
	\marginnote{Link to the original position of \cref{#1}, p. \pageref{#1}}
	}{
	\marginnote{{Proof\,p.\,{\pageref{app:#1}}}}
	}
	}{}
}

\newcommand{\applabel}[1]{$\phantomsection\label{app:#1}$}

%% file: macros/macros-diagrams.tex
\usepackage{tikz}
\usetikzlibrary{calc}
\usetikzlibrary{matrix,arrows}
\usetikzlibrary{decorations.shapes}
\usetikzlibrary{decorations.text}
\usetikzlibrary{decorations.pathmorphing}
\usetikzlibrary{decorations.markings}
\usetikzlibrary{positioning}

\tikzset{
node distance=1.3cm, auto,
every node/.style={font=\scriptsize },
ocenter/.style={baseline={([yshift=-.5ex, xshift=-.5ex]current bounding box)}},  
labelBeginAbove/.style={postaction={decorate,decoration={markings,mark=at position 0 with {\node[inner sep= 0.6pt, above=1pt]{\tiny #1};}} } },
labelBeginBelow/.style={postaction={decorate,decoration={markings,mark=at position 0 with {\node[inner sep= 0.6pt, below=1pt]{\tiny #1};}}}},
labelEndAbove/.style={postaction={decorate,decoration={markings,mark=at position 1 with {\node[inner sep= 0.6pt, above=1pt]{\tiny #1};}}}},
labelEndBelow/.style={postaction={decorate,decoration={markings,mark=at position 1 with {\node[inner sep= 0.6pt, below=1pt]{\tiny #1};}}}},
labelEndRight/.style={postaction={decorate,decoration={markings,mark=at position 1 with {\node[inner sep= 0.6pt, right=1pt]{\tiny #1};}}}},
labelEndLeft/.style={postaction={decorate,decoration={markings,mark=at position 1 with {\node[inner sep= 0.6pt, left=1pt]{\tiny #1};}}}}
}


%% file: local-macros.tex
\newcommand{\NoteProofNoComma}[1]{\withproofs{\ifthenelse{\boolean{appendix}}{Originally at p. \pageref{#1}}{Proof at p.\,{\pageref{app:#1}}}}}

\newcommand{\threestate}[3]{(#1, #2, #3)}

\renewcommand{\evctx}{O}
\renewcommand{\topos}{\Rew{\symfont{pos}}}

\renewcommand{\deriv}{e}
\renewcommand{\derivtwo}{\deriv'}
\renewcommand{\derivthree}{\deriv''}
\renewcommand{\derivfour}{\deriv'''}

%% file: 00-abstract.tex
Wu's positive $\lambda$-calculus is a recent call-by-value $\lambda$-calculus with sharing coming from Miller and Wu's study of the proof-theoretical concept of focalization. Accattoli and Wu showed that it simplifies a technical aspect of the study of sharing; namely it rules out the recurrent issue of renaming chains, that often causes a quadratic time slowdown. 

In this paper, we define the natural abstract machine for the positive $\lambda$-calculus and show that it suffers from an inefficiency: the quadratic slowdown somehow reappears when analyzing the cost of the machine. We then design an optimized machine for the positive $\lambda$-calculus, which we prove efficient. The optimization is based on a new slicing technique which is dual to the standard structure of machine environments.

%% file: 01-introduction.tex
\section{Introduction}
\label{sect:intro}
The $\l$-calculus is a minimalistic abstract setting that does not come with a fixed implementation schema. This is part of its appeal as a theoretical framework. Yet, for the very same reason, many different implementation techniques have been developed along the decades. How to implement the $\l$-calculus is in fact a surprisingly rich problem with no absolute best answer. 

Recently, a proof theoretical study about focusing by Miller and Wu \cite{DBLP:conf/csl/0001W23} led to a new $\l$-calculus with sharing---Wu's positive $\l$-calculus \Lpos \cite{DBLP:conf/aplas/Wu23}---that provides a fresh perspective on, and an improvement of a recurrent efficiency issue as shown by Accattoli and Wu \cite{entics:14758}. The present paper studies a somewhat surprising fact related to the efficiency of the positive $\l$-calculus. 

Similar to the ordinary $\l$-calculus, Wu's calculus is a minimalistic abstract setting. Its sharing mechanism, however, decomposes $\beta$-reduction in micro steps and makes it closer to implementations than the ordinary $\l$-calculus. It is, in fact, almost an abstract machine. 

This work stems from the observation that the natural refinement of \Lpos as an abstract machine suffers from an inefficiency. Intuitively, the efficiency issue improved by \Lpos resurfaces at the lower level of machines. After pointing out the problem, we design an optimized abstract machine and prove it efficient, solving the issue. The adopted \emph{slicing} optimization is new and based on a dual of the standard environment data structure for abstract machines. We believe it to be an interesting new implementation schema. 

\paragraph{The Positive $\l$-Calculus.} In \cite{DBLP:conf/csl/0001W23}, Miller and Wu decorate focalized proofs for minimal intuitionistic logic  with proof terms. Minimal intuitionistic logic is the proof-theoretical counterpart of the $\l$-calculus via the Curry-Howard correspondence. Focalization is a technique that constrains the shape of proofs depending on a polarity assignment to atomic formulas. Miller and Wu show that uniformly adopting the negative polarity induces the ordinary $\l$-calculus, while the positive polarity induces an alternative syntax with (sub-term) sharing, where sharing is represented via $\letexp$-expressions or, equivalently, explicit substitutions. 

An important aspect of this new positive syntax is its \emph{compactness property}, also referred to as \emph{positive sharing}: the shape of terms is highly constrained, more than in most other calculi with $\letexp$-expressions or explicit substitutions. In particular, applications cannot be nested, arguments can only be variables, applications and abstractions are always shared, and variables cannot be shared.  Some of these constraints are also at work in Sabry and Felleisen's \emph{A-normal forms} \cite{DBLP:conf/lfp/SabryF92,DBLP:conf/pldi/FlanaganSDF93}, of which positive sharing can be seen as an even more constrained variant; see the introduction of Accattoli and Wu \cite{entics:14758} for extended discussions about similar formalisms. Intuitively, the constrained syntax somewhat forces a maximal sharing of sub-terms by ruling out redundant cases of sharing from the grammar for terms.

In \cite{DBLP:conf/aplas/Wu23}, Wu endows the positive syntax with rewriting rules. Because of the compact syntax, there is not much freedom for defining the rules: they have to be call-by-value, and they have to be micro-step, that is, the granularity of the substitution process is the one of abstract machines, that replace one variable occurrence at a time. Call-by-name or meta-level substitution on all occurrences of a variable are indeed ruled out because they do not preserve compactness. The outcome is the positive $\l$-calculus \Lpos. The difference between \Lpos and an abstract machine is only that \Lpos does not have rules searching for redexes.

\paragraph{Renaming Chains.} In \cite{entics:14758}, Accattoli and Wu show two things. Firstly, compactness rules out a recurrent issue in the study of sharing and abstract machines, namely \emph{renaming chains}. In general, when one adds to the $\l$-calculus an explicit substitution construct, noted here $\tm\esub\var\tmtwo$ (equivalent to $\letin\var\tmtwo\tm$, and sharing $\tmtwo$ for $\var$ in $\tm$), then one can have chains of shared variables of the following form: $\tm\esub{\var_1}{\var_2}\ldots \esub{\var_{n-1}}{\var_n}$. 

These renaming, or indirection chains, are a recurrent burden of sharing-based systems, as they lead to both time and space inefficiencies---typically a quadratic time slowdown---and need optimizations to avoid their creations, as done for instance by Sands et al. \cite{DBLP:conf/birthday/SandsGM02}, Wand \cite{DBLP:journals/lisp/Wand07}, Friedman et al. \cite{DBLP:journals/lisp/FriedmanGSW07}, and Sestoft \cite{DBLP:journals/jfp/Sestoft97}. In the literature on sharing and abstract machines, the issue tended to receive little attention, until Accattoli and Sacerdoti Coen focused on it \cite{DBLP:journals/iandc/AccattoliC17}. Removing renaming chains is also essential in the study of reasonable logarithmic space for the $\l$-calculus, see Accattoli et al. \cite{DBLP:conf/lics/AccattoliLV22}.

In the positive $\l$-calculus, variables cannot be shared. Therefore, renaming chains simply cannot be expressed, ruling out the issue. It is important to point out, however, that forbidding the sharing of variables requires tuning the rewriting rules by adding some meta-level renamings.

\paragraph{Useful Sharing.} Secondly, Accattoli and Wu show that the compactness and the removal of renaming chains of the positive $\l$-calculus have the effect of drastically simplifying \emph{useful sharing}, a sophisticated implementation technique introduced by Accattoli and Dal Lago in the study of reasonable time cost models \cite{DBLP:journals/corr/AccattoliL16}. Intuitively, useful sharing aims at preventing the useless unfolding of sharing, which is exactly what is achieved by compactness via the restriction of the grammar of terms. That is, positive sharing captures the essence of useful sharing.

To be precise, Accattoli and Wu show that compactness simplifies the \emph{specification} of useful sharing, but they do not provide precise cost analyses---that are an essential aspect of useful sharing---for the positive $\l$-calculus.

\paragraph{The Inefficiency.} The inception of the present work is exactly the desire to develop a cost analysis of \Lpos. Cost analyses are usually done via a study of an abstract machine for the calculus of interest. The somewhat surprising fact is that the natural abstract machine associated to the positive $\l$-calculus---first given here, and dubbed Natural POsitive Machine, or \emph{Natural POM}---is inefficient. The culprit is the mentioned meta-level renamings added to compensate for the fact that variables cannot be shared. The implementation of these renamings induces a quadratic (rather than linear) overhead in the number of $\beta$-steps. Essentially, the inefficiency of renaming chains reappears at the lower level of the Natural POM even if the chains themselves have disappeared. 

Useful sharing is meant to reduce the overhead from exponential to polynomial, so the inefficiency does not invalidate the value of \Lpos for specifying useful sharing. The literature however contains abstract machines for useful sharing having only a \emph{linear} overhead (by Accattoli and co-authors \cite{DBLP:conf/lics/AccattoliC15,DBLP:journals/scp/AccattoliG19,DBLP:conf/ppdp/AccattoliCGC19,DBLP:conf/lics/AccattoliCC21}), thus the Natural POM compares poorly, despite the fact that the positive $\l$-calculus is a better specification of useful sharing than those in the literature.

\paragraph{The Solution: Sub-Term Property and Slices.} We then design an optimized abstract machine, the \emph{Sliced POM}, that solves the issue and recovers linearity in the number of $\beta$-steps. To give a hint of the solution, we need to say a bit more about the problem. 
Complexity analyses of abstract machines are based on their crucial \emph{sub-term property}, stating that all the terms duplicated along the run of the machine are sub-terms of the initial term. The property allows one to bound the cost of each transition using the initial term, which is essential in order to express the cost of the run as a function of the size of the initial term. 

Now, the Natural POM does verify the sub-term property \emph{with respect to duplications}; the source of the inefficiency for once is not the duplication process. The source is nonetheless related to a lack of sub-term property \emph{for renamings}: the additional meta-level renamings act over a scope that might not be a sub-term of the initial term. The solution amounts to \emph{slicing} such scopes in slices that are sub-terms of the initial term, and in noting that each meta-level renaming is always confined to exactly one slice. The slicing technique is managed very easily, via an additional basic data structure, the slice stack. 

\paragraph{Slices $vs$ Environments.} A pleasant aspect is that the slicing stack can be seen as \emph{dual} to the environment data structure used to manage sharing: 
\begin{itemize}
\item an environment entry $\esub\var\tm$ stores the delayed substitution of $\tm$ for $\var$ waiting for the evaluation of active term to expose occurrences of $\var$ to replace;
\item an entry $\sliceentry\tm \var$ of the slice stack waits for the active term to become a variable $\varthree$ to be substituted for $\var$ in the slice $\tm$. 
\end{itemize}
This duality suggests that the slicing technique exposes a new natural structure. The meta-level renamings of the calculus are similar to the ones generated by $\beta$-redexes. In traditional settings with sharing, indeed, the optimizations to remove renaming chains act on the environment. We find it interesting that positive sharing \emph{disentangles} meta-level renamings from the usual substitution process, and manages them via a sort of dual mechanism.

\paragraph{On the Value of Positive Sharing.} The reader might wonder what the value of positive sharing is, given that the issue that it is supposed to solve does reappear at the level of machines and forces the design of a new solution. In our opinion, the value of positive sharing is in re-structuring the study of sharing techniques. The study by Accattoli and Wu \cite{entics:14758} suggests that positive sharing \emph{removes} the issue of renaming chains, but our work provides a refined picture. Namely, positive sharing  enables a neater theory of sharing, disentangling renaming chains from the treatment of explicit substitutions, and encapsulating their inefficiency at the lower level of implementations choices. Additionally, it brings to the fore the concept of slice stack, which we believe is an interesting addition to the theory of implementations of the $\l$-calculus. 

\paragraph{OCaml Implementation.} For lack of space we overview in \withproofs{\refapp{app-implementation}}\withoutproofs{Appendix A of the technical report \cite{techreport}} an implementation in OCaml of the
Sliced POM, to be found at~\url{https://github.com/sacerdot/PositiveAbstractMachine}.
The implementation is meant to provide further evidence on the cost of the atomic operations of the machine.
Moreover, it allows the interested reader to enter ordinary $\lambda$-terms and see how they are transformed into positive terms (via the transformation studied by Accattoli and Wu in~\cite{entics:14758}) and then run by the Sliced POM.

The implementation is a prototype: it does not attempt to further optimize space usage or to recover garbage memory, nor does it optimize the code that is unrelated to running the machine (\eg the pretty-printing of machine states).

\withproofs{\paragraph{Proofs.} Omitted proofs are in the Appendix.}
\withoutproofs{\paragraph{Proofs.} Most proofs are omitted. They can be found in the appendix of the technical report on arXiv~\cite{techreport}.}

%% file: 02-positive_calculus.tex
\section{The Positive $\l$-Calculus}
\label{sect:positive}

\input{\figurespath/figure-calculus}
In this section, we present Wu's positive $\l$-calculus \cite{DBLP:conf/aplas/Wu23}. Precisely, we adopt the \emph{explicit open} variant from Accattoli and Wu \cite{entics:14758}. For simplicity, we simply refer to it as the positive $\l$-calculus, and note it \Lpos. We  depart slightly from the presentation in \cite{entics:14758}, omitting the garbage collection rule because it can always be postponed, as shown in \cite{entics:14758}, and changing some notations. The definition is in \reffig{explicit-positive}.

\paragraph{Terms.} The positive $\l$-calculus uses $\letexp$-expressions, similarly to Moggi's \cbv calculus \cite{Moggi88tech,DBLP:conf/lics/Moggi89}. We do however write a $\letexp$-expression $\letin\var\tmtwo\tm$ as a more compact  \emph{explicit substitution} $\tm\esub{\var}{\tmtwo}$ (\ES for short), which binds $\var$ in $\tm$. Moreover, for the moment our $\letexp$/ES does not fix an order of evaluation between $\tm$ and $\tmtwo$, in contrast to many papers in the literature (\eg Levy et al. \cite{DBLP:journals/iandc/LevyPT03}) where $\tmtwo$ is evaluated first. The evaluation order shall be fixed in the next section.

While we borrow the terminology \emph{explicit substitution} from the seminal work of Abadi et al. \cite{DBLP:journals/jfp/AbadiCCL91}, the way we employ the construct is deeply different, since the theory of ESs has progressed considerably since that first work. In particular, we do use variable names instead of de Bruijn indices, and our ESs do not move through the structure of the term, instead they act \emph{at a distance} (explained below). The positive $\l$-calculus is rather reminiscent of Sabry and Felleisen's \emph{A-normal forms} \cite{DBLP:conf/lfp/SabryF92,DBLP:conf/pldi/FlanaganSDF93}, or of the calculi for call-by-need by Lunchbury \cite{DBLP:conf/popl/Launchbury93} and Sestoft \cite{DBLP:journals/jfp/Sestoft97}, of which it can be thought as a more constrained variant. 

In fact, ESs are here used in a slightly unusual way even with respect to more recent work on ESs. In \Lpos, as in the $\l$-calculus, there are only three constructors, variables, applications, and abstractions. There are however, various differences, namely:
\begin{itemize}
\item Applications have either shape $\vartwo\varthree$ or $(\la\vartwo\tm)\varthree$, that is, arguments can only be variables and the left sub-term cannot be an application;
\item Applications and abstractions are dubbed \emph{bites} (following the terminology of Accattoli et al. \cite{DBLP:conf/ppdp/AccattoliCGC19}) and are always shared, that is, standalone applications and abstractions are not allowed by the grammar. They can only be introduced by the ESs constructs $\esub{\var}{\vartwo\varthree}$, $\esub{\var}{(\la\vartwo\tm)\varthree}$, and $\esub{\var}{\la\vartwo\tmtwo}$;
\item Positive sharing is peculiar as positive terms are \emph{not} shared in general, that is, $\tm\esub\var\tmtwo$ is not a positive term. There is no construct for sharing variables or applications/abstractions with top-level sharing, \ie $\tm\esub\var\vartwo$ or $\tm\esub\var{\vartwo\varthree\esub\vartwo{\la\varfour\tmtwo}}$ are not terms of \Lpos. In particular, the absence of $\tm\esub\var\vartwo$ is what forbids renaming chains.
\end{itemize}
The set of free variables of a term $\tm$ is denoted by $\fv{\tm}$ and it is defined as expected. Terms are identified up to $\alpha$-renaming. We use $\tm \isub{\var}{\vartwo}$ for the capture-avoiding substitution of $\vartwo$ for each free occurrence of $\var$ in $\tm$; in this paper we never need the more general operation $\tm \isub{\var}{\tmtwo}$ substituting terms. The meta-level renamings mentioned in the introduction, to be used to compensate for the absence of $\tm\esub\var\vartwo$, shall be instances of $\tm \isub{\var}{\vartwo}$.

\paragraph{Open Setting.} Evaluation in \Lpos shall be open in the sense promoted by Accattoli and Guerrieri \cite{DBLP:conf/aplas/AccattoliG16}, that is, it does not go under abstraction (also referred to as \emph{weak}) and terms \emph{can} be open (but do not have to). As discussed at length by Accattoli and Guerrieri, this is a more general framework than the standard for functional programming languages of weak evaluation and \emph{closed} terms. The increased generality enables the developed theory to scale up to evaluation under abstraction, which is needed to model proof assistants, by iterating evaluation under abstraction---this is done for instance by \gregoire and Leroy for Coq \cite{DBLP:conf/icfp/GregoireL02}---because the body of abstractions cannot be assumed to be closed. See \cite{DBLP:conf/aplas/AccattoliG16} for more discussions.

\paragraph{Open Contexts.}   Contexts are terms with exactly one occurrence of the \emph{hole} $\ctxhole$, an additional constant, standing for a removed sub-term. We shall use various notions of contexts. The most general ones in this paper are \emph{open contexts} $\evctx$, that are simply lists of ESs.
The main operation about contexts is \emph{plugging} $\evctxp{\tm}$ where the hole $\ctxhole$ in context 
$\evctx$ is replaced by $\tm$. Plugging, as usual with contexts, can
capture variables---for instance $(\ctxhole\esub\var\bite)\ctxholep\var = \var\esub\var\tmtwo$. 

The domain $\dom\evctx$ of a context is the set of variables possibly captured by $\evctx$ (i.e., on which $\evctx$ has an ES scoping over $\ctxhole$); example: setting $\evctx \defeq \ctxhole\esub\var{\la\vartwo\vartwo\esub{\vartwo'}{\vartwo''\vartwo''}}\esub\varthree{\varthree'\varthree''}$ one obtains $\dom{\evctx} = \set{\var,\varthree}$. When $\var\in\dom\evctx$, we also use the notation $\evctx(\var)$ to denote the bite associated to $\var$ in $\evctx$; in the example, $\evctx(\varthree) = \varthree'\varthree''$.

 As it is immediately seen from the grammar of positive terms, every positive term $\tm$ can be written uniquely as $\evctxp\var$ for some $\var$ and $\evctx$, with $\evctx$ possibly capturing $\var$. If $\tm=\evctxp\var$ then $\var$ is referred to as the \emph{head variable} of $\tm$.

\paragraph{Rewriting Rules.} There are two rewriting rules, following the \emph{at a distance} style promoted by Accattoli and Kesner \cite{DBLP:conf/csl/AccattoliK10}, which involves contexts in the definition of the rules, even before the contextual closure. The rules names come from the connection with linear logic proof nets, which is omitted here.

The multiplicative rule $\tom$ reduces a shared $\beta$-redex. The rule is forced to decompose the body of the abstraction as $\evctxp\varthree$ in order to write the reduct. The point is that the simpler rule $\tm\esub\var{(\la\vartwo\tmtwo)\varfour} \rtom \tm\esub\var{\tmtwo\esub\vartwo\varfour}$ does not respect the grammars of the positive $\l$-calculus, since ESs such as $\esub\var\tmtwo$ and $\esub\vartwo\varfour$, as well as their nesting, are forbidden. Let us show an example of multiplicative step stressing the action of renamings:
\begin{center}
$\begin{array}{l\colspace c\colspace l\colspace \colspace l}
\varthree\esub\var{\vartwo\varthree}\esub\varthree{(\la\varfour\varthree'\esub{\var'}{\varfour\varthree'})\vartwo'} 
& \tom&
\varthree\esub\var{\vartwo\varthree}\isub\varthree{\varthree'}\esub{\var'}{\varfour\varthree'}\isub\varfour{\vartwo'}
\\ & =&
\varthree'\esub\var{\vartwo\varthree'}\esub{\var'}{\vartwo'\varthree'}

\end{array}$
\end{center}

The exponential rule $\toe$ simply replaces an applied variable with the associated abstraction. Note that arguments, that are variables, are never replaced, because their replacement would---once more---step out of the grammar of the positive $\l$-calculus.

Both rules are closed by open contexts and together form the rewriting relation $\topos$ of the positive $\l$-calculus.

\paragraph{Translating $\l$-Terms.} In \cite{entics:14758}, Accattoli and Wu show how to translate $\l$-terms to positive terms in a way that induces a simulation of call-by-value evaluation. We refer the interested reader to their work, because in this paper we only deal with positive terms.

\paragraph{Diamond.} The defined calculus is non-deterministic. Consider for instance $\tm \defeq \varthree\esub\var{\vartwo\vartwo}\esub\varthree{(\la\varfour\varfour)\vartwo'}\esub\vartwo{\la{\var'}\tmtwo}$. One has for instance the following diagram:
\begin{center}
\begin{tikzpicture}[ocenter]
		\node at (0,0)[align = center](source){\normalsize $\tm$};
		\node at (source.center)[below = 20pt](source-down){\normalsize $\varthree\esub\var{(\la{\var'}\tmtwo)\vartwo}\esub\varthree{(\la\varfour\varfour)\vartwo'}\esub\vartwo{\la{\var'}\tmtwo}$};
		\node at (source.center)[right = 130pt](source-right){\normalsize $\vartwo'\esub\var{\vartwo\vartwo}\esub\vartwo{\la{\var'}\tmtwo}$};
		\node at (source-right|-source-down)(target){\normalsize $\vartwo'\esub\var{(\la{\var'}\tmtwo)\vartwo}\esub\vartwo{\la{\var'}\tmtwo}$};
				
		\draw[->](source) to node[above] {\scriptsize $\msym$} (source-right);
		\draw[->](source) to node[left] {\scriptsize $\esym$} (source-down);
		
		\draw[->, dotted](source-down) to node[above] {\scriptsize $\msym$} (target);
		\draw[->, dotted](source-right) to node[right] {\scriptsize $\esym$} (target);

	\end{tikzpicture}
\end{center}
The calculus however is confluent, and even more than confluent, it has the diamond property. According to Dal Lago and Martini \cite{DBLP:journals/tcs/LagoM08}, a relation $\to$ is \emph{diamond} if $\tmtwo_1 \lto \tm \to \tmtwo_2$  imply $\tmtwo_1 = \tmtwo_2$ or $\tmtwo_1 \to \tmthree \, \lto \tmtwo_2$ for some $\tmthree$. The diamond property expresses a relaxed form of determinism, since it states that different choices cannot change the result \emph{nor the length of evaluation sequences} (note that the diagram closes in either zero or one steps on both sides).
\begin{theorem}[Positive diamond, \cite{entics:14758}]
Relation $\topos$ is diamond.\label{thm:lppos-diamond}
\end{theorem}

\paragraph{Sub-Term Property.} When the substitution process is decomposed in micro steps, usually it is possible to bound the cost of each duplication along an evaluation sequence using the size of the initial term of the sequence---this is the sub-term property. It is crucial in order to analyze the cost of evaluation sequences as a function of the size of the initial term and the number of steps, since the main danger of excessive cost usually comes from duplications. The property does not hold, for instance, for the ordinary $\l$-calculus (see Accattoli \cite[Section 3]{DBLP:journals/lmcs/Accattoli23}) that relies on meta-level (rather than micro-step) substitution. The positive $\l$-calculus has the sub-term property, that is expressed for values, because they are what is duplicated by the exponential rule.

\begin{lemma}[Sub-term property]
Let $\tm \topos^* \tmtwo$ be a reduction sequence. Then $\size{\la\var\tmthree}\leq\size\tm$ for every bite $\la\var\tmthree$ duplicated by a $\esym$-step of the sequence.\label{l:sub-term-calculus}
\end{lemma}

\begin{proof}
Formally, the proof is a straightforward induction on the length of the reduction sequence, by looking at the last step and using the \ih The following informal observations however are probably enough. Evaluation  duplicates variables (in $\tom$) and abstractions (in $\toe$). The substitution of variables cannot change the size of abstractions. For abstractions, note that replaced variables are out of abstractions, and open contexts never enter abstractions, so that all abstractions of the sequence can be traced back to $\tm$ (up to $\alpha$-renaming).\qed
\end{proof}
\section{The Right Strategy} 
\label{sect:right-strategy}
Usually, abstract machines are deterministic and implement a deterministic evaluation strategy. Therefore, in this section, we define a deterministic strategy for the positive $\l$-calculus and prove its basic properties. 

We adopt the right(-to-left) strategy $\tor$ that picks redexes from right to left. It is a standard approach but it turns out that defining it in \Lpos is a bit tricky, because of $\esym$-steps, where two ESs interact at a distance.

\paragraph{Redex Positions.} For calculi at a distance, the notion of position of a redex---which is mandatory to determine the rightmost redex---might not be as expected. For us, a position in a term is simply an open context.

\begin{definition}[Redex positions]
When taking into account the contextual closure, $\msym$-steps and $\esym$-steps have the following shapes: 
\begin{center}
$\begin{array}{l\colspace l\colspace r\colspace l\colspace lll}
\tm &=& \evctxtwop{\tmthree\esub\var{(\la\vartwo{\evctxp{\varthree}})\varfour}}
				    & \tom &
\evctxtwop{\evctxp{\tmthree\isub{\var}{\varthree}}\isub{\vartwo}{\varfour}} & = & \tmtwo
\\[3pt]
\tm &=&\evctxtwop{ \evctxp{\tmthree\esub\var{\vartwo\varthree}}
\esub\vartwo{\la\varfour\tmfour}}
				    & \toe &
\evctxtwop{\evctxp{\tmthree\esub{\var}{(\la\varfour\tmfour)\varthree}}\esub\vartwo{\la\varfour\tmfour}} & = & \tmtwo
\end{array}$
\end{center}
The position of the $\msym$-step is simply given by the surrounding context $\evctxtwo$. The position of the $\esym$-step is the context $\evctxtwop{ \evctx\esub\vartwo{\la\varfour\tmfour}}$.
\end{definition}
The rationale behind the position of $\esym$-steps is the idea that they are triggered when one finds the variable to substitute, and not when one finds the abstraction---this is indeed how abstract machines work. This approach is  standard in the literature about ESs at a distance, see \eg Accattoli et al. \cite{DBLP:conf/popl/AccattoliBKL14}.  

Example: in $\var\esub\var{\vartwo\varthree}\esub{\var'}{\vartwo'\varthree}\esub{\vartwo'}{\la{\varfour'}\tmthree}\esub\vartwo{\la\varfour\tmtwo}$ there are two redexes (on $\vartwo$ and $\vartwo'$) and the rightmost one is on $\vartwo'$, despite the ES on $\vartwo'$ occurring to the left of the one on $\vartwo$.

\paragraph{Right Contexts.} We specify the strategy via the notion of right contexts, itself specified via the notion of applied free variable. In fact, there are two dual ways of defining right contexts. We present them both and prove their equivalence.

\begin{definition}[Applied free variables, right contexts]
The set of applied (and out of abstraction bodies) free variables $\afv\evctx$  of an open context $\evctx$ is defined as:\label{def:right-ctx}
\begin{center}
$\begin{array}{rll\colspace |\colspace rll}
\multicolumn{6}{c}{\textsc{Set of applied free variables}}
\\
\afv\ctxhole & \defeq &\emptyset
&
\afv{\evctx\esub\var{\vartwo\varthree}} & \defeq &(\afv\evctx\setminus\set{\var})\cup \set\vartwo
\\
\afv{\evctx\esub\var{(\la\vartwo\tm)\varthree}} & \defeq &\afv\evctx\setminus\set{\var}
&
\afv{\evctx\esub\var{\la\vartwo\tm}} & \defeq &\afv\evctx\setminus\set{\var}

\end{array}$
\end{center}
The two definitions of right contexts (we use on purpose the same meta-variable, since they shall be proved equivalent right next) are given by:
\begin{center}
$\begin{array}{lllll}
\multicolumn{3}{c}{\textsc{Outside-in right contexts}}
\\
\rctx & \grameq &\ctxhole \mid \rctx\esub\var{\vartwo\varthree}  \mid \rctx\esub\var{\la\vartwo\tmtwo} \text{ if }\var\notin\afv\rctx
\\[4pt]
\multicolumn{3}{c}{\textsc{Inside-out right contexts}}
\\
 \rctx & \grameq &\ctxhole \mid \rctxp{\ctxhole\esub\var{\vartwo\varthree}} \text{ if }\rctx(\vartwo)\neq \la\varfour\tm\ \mid \rctxp{\ctxhole\esub\var{\la\vartwo\tmtwo}}
\end{array}$
\end{center}
\end{definition}

\begin{toappendix}
\begin{lemma}
\NoteProof{l:right-ctxs-coincide}
Outside-in and inside-out right contexts coincide. \label{l:right-ctxs-coincide}
\end{lemma}
\end{toappendix}
Because of the lemma, we only speak of right contexts, and adopt the most convenient definition in each case.


\paragraph{Right Strategy.} We now have all the ingredients to define the right strategy and prove its basic properties.

\begin{definition}[Right strategy]
A $\topos$ step is \emph{right} when its position is a right context. The right strategy $\tor$ reduces at each step a right redex, if any. We write $\tm\torm\tmtwo$ (resp. $\tm\tore\tmtwo$) for a right $\msym$-step (resp. $\esym$-step).
\end{definition}


\begin{toappendix}
\begin{lemma}[Basic properties of the right strategy]
\label{l:right-strat-properties}\NoteProof{l:right-strat-properties}\hfill
\begin{enumerate}
\item \emph{Determinism}: if $\tm \tor \tmfirst$ and $\tm \tor \tmsec$ then $\tmfirst=\tmsec$;
\item \emph{No premature stops}: if $\tm\topos \tmtwo$ then $\tm \tor \tmthree$ for some $\tmthree$.
\end{enumerate}
\end{lemma}
\end{toappendix}

%% file: figures/figure-calculus.tex
\begin{figure}[t!]
\centering
\fbox{
		\begin{tabular}{c}
				$
				\begin{array}{rrll\colspace\colspace\colspace rrll }
		\textsc{Bites } & \bite, \bitetwo & \grameq& \vartwo\varthree \mid \la\vartwo\tmtwo \mid (\la\vartwo\tmtwo)\varthree
		&
		\textsc{Evaluation Ctxs } & \evctx & \grameq &\ctxhole \mid \evctx\esub\var\bite
				\\[3pt]
		\textsc{Terms } & \tm, \tmtwo, \tmthree & \grameq &\var \mid \tm \esub\var\bite
				\end{array}$
			\\[10pt]
			\hline
				$\begin{array}{c}
					\textsc{Root reduction rules}
					\\[4pt]
					\begin{array}{r\colspace rll}
  				    \textsc{Multiplicative} & \tm\esub\var{(\la\vartwo{\evctxp{\varthree}})\varfour}
				    & \rtom &
\evctxp{\tm\isub{\var}{\varthree}}\isub{\vartwo}{\varfour}
					\\[3pt]
				    \textsc{Exponential} & \evctxp{\tm\esub\var{\vartwo\varthree}}
\esub\vartwo{\la\varfour\tmtwo}
				    & \rtoe &
\evctxp{\tm\esub{\var}{(\la\varfour\tmtwo)\varthree}}\esub\vartwo{\la\varfour\tmtwo}
\\
&&&\text{ with }\vartwo\notin\dom\evctx
					\end{array}
					\end{array}$
					\\[17pt]
					\hline
					$\begin{array}{c\colspace |\colspace c}

					\begin{array}{rl}
					\textsc{Ctx closure}
					&
					\RightLabel{$\asym{\in}\{\msym,\esym\}$}
					\AxiomC{$\tm \rootRew{\asym} \tm'$}
					\UnaryInfC{$\evctxp\tm \Rew{\asym} \evctxp{\tm'}$}
					\DisplayProof
					\end{array}
					&
					\textsc{Notation} \ \ 					\topos    \defeq  \tom \cup \toe					
				\end{array}$
				\end{tabular}
			}
				\caption{The positive $\l$-calculus \Lpos.}
				\label{fig:explicit-positive}	
		\end{figure}

%% file: 03-natural_machine.tex
\section{A Natural but Inefficient Positive Machine}
\label{sect:failure}
In this section, we give the natural abstract machine implementing the right strategy of the previous section, obtained by adding a basic mechanism for searching redexes, discuss the (in)efficiency of the machine, and explain how to modify it as to make it efficient. The tone is slightly informal. The next sections shall formally define and study the modified machine. 

Typical abstract machines for the $\l$-calculus in the literature are the Krivine abstract machine (KAM) or Felleisen and Friedman's CEK machine \cite{DBLP:conf/ifip2/FelleisenF87} that use many environments, closures, and never $\alpha$-rename. Here we rely on a different and simpler approach, having only one global environment (represented as a context), no closures, and using $\alpha$-renaming. For comparisons between the two approaches, see Accattoli and Barras \cite{DBLP:conf/ppdp/AccattoliB17}.

\input{figures/figure-Natural-POM}
\paragraph{The Natural POM.} States of the Natural POsitive Machine (Natural POM\footnote{Acronyms for abstract machines tend to end with AM for Abstract Machine. We avoided PAM, however, because there already is a PAM (Pointer Abstract Machine) in the literature, introduced by Danos et al. \cite{DBLP:conf/lics/DanosHR96}.
}) are pairs $\tm\lhd\rctx$ denoting a pointer $\lhd$ inside the structure of the term $\tmtwo \defeq \rctxp\tm$ represented by the state. The pointer represents the current position of the machine over $\tmtwo$. Initially, the pointer is at the rightmost position, that is, initial states have shape $\tm\lhd\ctxhole$. The Natural POM has the four transitions in \reffig{natural-POM}. The differences with respect to the positive $\l$-calculus are that:
\begin{itemize}
\item \emph{Search}: there are two search transitions $\seasym_1$ and $\seasym_2$ to move the pointer, in order to search for redexes, that move $\lhd$ right-to-left (hence the symbol);
\item \emph{Names}: there is a more controlled management of $\alpha$-conversion, which is performed only on the exponential step, when copying abstractions.
\end{itemize}
Transition $\seasym_1$ moves the abstraction from the active code to the right context. Transition $\seasym_2$ moves $\esub\var{\vartwo\varthree}$ when $\vartwo$ is a free variable (that is, $\vartwo\notin\dom\rctx$) or when it is bound by an ES in $\rctx$ but not one that contains an abstraction, so that the application $\vartwo\varthree$ does not give rise to a multiplicative redex. Note that these two cases are exactly the inside-out definition of right contexts in \refdef{right-ctx}.

\paragraph{Example of Run.} As an example of execution of the Natural POM, we consider the first few transitions of the infinite run for the positive representation $ \var\esub\var{\vartwo\vartwo}\esub\vartwo{\la\varthree\varfour\esub\varfour{\varthree\varthree}}$ of the paradigmatic looping $\l$-term $\Omega \defeq (\la\vartwo\vartwo\vartwo)(\la\varthree\varthree\varthree)$:
\begin{center}
$\begin{array}{l\colspace lllllllll}
\textsc{Active code} & &\textsc{Right ctx}
\\\hhline{---}
 \var\esub\var{\vartwo\vartwo}\esub\vartwo{\la\varthree\varfour\esub\varfour{\varthree\varthree}} & \lhd & \ctxhole 
 & \tomachseaone
\\
 \var\esub\var{\vartwo\vartwo} & \lhd & \esub\vartwo{\la\varthree\varfour\esub\varfour{\varthree\varthree}}
& \tomache
\\
\var\esub\var{(\la{\varthree'}{\varfour'}\esub{\varfour'}{\varthree'\varthree'})\vartwo} & \lhd & \esub\vartwo{\la\varthree\varfour\esub\var{\varthree\varthree}}
& \tomachm
\\
\varfour'\esub{\varfour'}{\vartwo\vartwo} & \lhd & \esub\vartwo{\la\varthree\varfour\esub\var{\varthree\varthree}}
& \tomache&\ldots
\\[2pt]
\end{array}$
\end{center}

\paragraph{Basics of the Complexity of Abstract Machines.} The problem with the Natural POM concerns the complexity of its overhead with respect to the calculus. Let us first recall the basics of the topic. For abstract machines, the complexity of the overhead for a machine run $\run$ of initial term $\tm$ is measured with respect to two parameters: the number of steps of the underlying calculus and the size $\size\tm$ of the initial term. A large number of machines has complexity linear in both parameters---shortened to \emph{bi-linear}---as first shown by Accattoli et al. \cite{DBLP:conf/icfp/AccattoliBM14}, and repeatedly verified after that for even more machines \cite{DBLP:conf/lics/AccattoliC15,DBLP:conf/aplas/AccattoliBM15,DBLP:conf/ppdp/AccattoliCGC19,DBLP:journals/scp/AccattoliG19,DBLP:conf/lics/AccattoliCC21}. The bi-linearity of the overhead then becomes a design principle, or, when it fails, a strong indication of an inefficiency, and that the machine can be improved.

The key property enabling complexity analyses of machines is the sub-term property already discussed for \Lpos. It ensures that the size of states cannot grow more than the size of the initial term at each transition. In turn, this fact usually implies that the time cost is also linearly bounded.

\paragraph{The Inefficiency of the Natural POM.} The natural POM inherits the sub-term property from the calculus. Its inefficiency is related to the multiplicative transition $\tomachm$, namely to the meta-level renaming $\tm\isub\var\varthree$ (referring to \reffig{natural-POM}). Meta-level substitutions are possibly costly. Usually, the danger is the time spent in making copies of the term to duplicate, which might be big and/or because many copies of it might be required. The size of the term to duplicate is not the problem here, since the involved term is a simple variable $\varthree$.

The culprit actually is the \emph{size of the scope} over which the renaming can take place, rather than the number of copies. There are two renamings. The renaming $\isub\vartwo\varfour$ is harmless: its scope seems to be $\evctxp{\tm\isub\var\varthree}$ but in fact $\vartwo$ can occur only in $\evctxp\varthree$ (which is the body of the abstraction of $\vartwo$ in the source state of the transition), and the sub-term property guarantees that the size of $\evctxp\varthree$ is bound by the size of the initial term. Therefore, one can propagate $\isub\vartwo\varfour$ on $\evctxp\varthree$ before computing the full reduct, staying within a linear cost.

For the renaming $\tm\isub\var\varthree$, however, there is in general no connection between $\tm$ and the initial term. For the first multiplicative step of the run, $\tm$ is a sub-term of the initial term. But the step itself re-combines $\evctx$ and $\tm$ as to create a new term unrelated to the initial one. Consider for instance the following run, where we assume that $\evctx$ captures $\varthree$ and that $\rename{(\la\vartwo(\evctxp{\varthree}\esub{\varthree'}\bite))} = \la{\vartwo'}(\evctxtwop{\varthree'}\esub{\varthree''}{\bite'})$:
\begin{equation}
\begin{array}{l\colspace lllllllll}
&\textsc{Active code} & &\textsc{Right ctx}
\\\hline
& \tm\esub\var{\var'\varfour}\esub{\var'}{\la\vartwo(\evctxp{\varthree}\esub{\varthree'}\bite)} & \lhd & \ctxhole 
\\[2pt]
\tomachseaone &
\tm\esub\var{\var'\varfour} & \lhd & \esub{\var'}{\la\vartwo(\evctxp{\varthree}\esub{\varthree'}\bite)}
\\[2pt]
\tomache &
\tm\esub\var{(\la{\vartwo'}(\evctxtwop{\varthree'}\esub{\varthree''}{\bite'}))\varfour} & \lhd & \esub{\var'}{\la\vartwo(\evctxp{\varthree}\esub{\varthree'}\bite)}
\\[2pt]
\tomachm &
\evctxtwop{\tm\isub\var{\varthree'}}\esub{\varthree''}{\bite'}\isub{\vartwo'}\varfour & \lhd & \esub{\var'}{\la\vartwo(\evctxp{\varthree}\esub{\varthree'}\bite)}
\end{array}
\label{eq:inefficiency}
\end{equation}
Now, if the bite $\bite'$ is a $\beta$-redex then the next transition is multiplicative and it is going to rename over $\evctxtwop{\tm\isub\var{\varthree'}}$ which  is not a sub-term of the initial term.

\paragraph{How Big is the Inefficiency?} When the time cost depends on something that is not a sub-term of the initial term, things can easily escalate up to exponential costs, as in the paradigmatic case of size explosion, see Accattoli \cite[Section 3]{DBLP:journals/lmcs/Accattoli23}. Luckily, here things do not go sideways, there is only a mild inefficiency. The key point is that the sub-term property for duplications \emph{does hold}: it ensures that the size of the whole state grows bi-linearly with the number of transitions, so there is no exponential growth. The problematic renaming---and thus each multiplicative transition---might then have to scan a scope that is at worst bi-linear in the length of the preceding run and the size of the initial term. A standard argument then gives a quadratic bound (in the number of steps and the size of the initial term) on the global cost of all multiplicative steps.

Let us give an example. We use a diverging term because it is the simplest example showcasing the phenomenon. We define the positive representation of $\Omega_3 \defeq \delta_3\delta_3$ where $\delta_3 \defeq \la\var\var(\var\var)$. Let $\tau_3 \defeq \var\esub\var{\vartwo\varthree}\esub\varthree{\vartwo\vartwo}$ and let $\tau_3' \defeq  \var'\esub{\var'}{\vartwo'\varthree'}\esub{\varthree'}{\vartwo'\vartwo'}$ and so on. The analogous of $\Omega_3$ is $\tau_3\esub\vartwo{\la\vartwo\tau_3}$, that runs as follows:
\begin{center}
$\begin{array}{l\colspace lllllllll}
\textsc{Active code} & &\textsc{Right ctx}
\\\hhline{---}
 \var\esub\var{\vartwo\varthree}\esub\varthree{\vartwo\vartwo} \esub\vartwo{\la\vartwo\tau_3}& \lhd & \ctxhole 
& \tomachseaone
\\[2pt]
\var\esub\var{\vartwo\varthree}\esub\varthree{\vartwo\vartwo} & \lhd & \esub\vartwo{\la\vartwo\tau_3} 
&\tomache
\\[2pt]
\var\esub\var{\vartwo\varthree}\esub\varthree{(\la{\vartwo'}\tau_3')\vartwo} & \lhd & \esub\vartwo{\la\vartwo\tau_3}
&=
\\[2pt]
 \var\esub\var{\vartwo\varthree}\esub\varthree{(\la{\vartwo'}\var'\esub{\var'}{\vartwo'\varthree'}\esub{\varthree'}{\vartwo'\vartwo'})\vartwo} & \lhd & \esub\vartwo{\la\vartwo\tau_3}
&\tomachm
\\[2pt]
\var\esub\var{\vartwo\var'}\esub{\var'}{\vartwo\varthree'}\esub{\varthree'}{\vartwo\vartwo} & \lhd & \esub\vartwo{\la\vartwo\tau_3}
& \tomache
\\[2pt]
\var\esub\var{\vartwo\var'}\esub{\var'}{\vartwo\varthree'}\esub{\varthree'}{(\la{\vartwo''}\tau_3'')\vartwo} & \lhd & \esub\vartwo{\la\vartwo\tau_3}
& \tomachm
\\[2pt]
\var\esub\var{\vartwo\var'}\esub{\var'}{\vartwo\var''}\esub{\var''}{\vartwo\varthree''}\esub{\varthree''}{\vartwo\vartwo} & \lhd & \esub\vartwo{\la\vartwo\tau_3} 
&\tomache&\ldots

\end{array}$
\end{center}
It is clear that the scopes of the renamings keep growing, even if the number of renamed occurrences by each renaming is constant.

\paragraph{Removing the Inefficiency: Slices.} An observation about the run \refeq{inefficiency} above suggests how to improve the situation. The idea is to delay the merging of $\tm$ and $\evctxtwop{\varthree'}$ as $\evctxtwop{\tm\isub\var{\varthree'}}$ (similarly to how ESs delay meta-level substitutions) by putting the pair $(\tm,\var)$---dubbed \emph{slice} and that shall actually be denoted with $\sliceentry\tm\var$---in a new stack storing delayed merges, and keeping as active code $\evctxtwop{\varthree'}$. The observation is that if $\bite'$ is a $\beta$-redex and generates a renaming of $\varthree''$ then one only needs to inspect $\evctxtwop{\varthree'}$---which, crucially, is a sub-term of the initial term---because $\varthree''$ cannot occur in $\tm$, given that it comes from the body of an abstraction out of $\tm$.

Before defining the Sliced POM, and prove that slices do solve the problem, we start over with a more formal approach to abstract machines.

%% file: figures/figure-Natural-POM.tex
\begin{figure}[t!]
\centering
\fbox{
\begin{tabular}{c}
						$\begin{array}{lll || c || lll llllll}
						\multicolumn{7}{c}{\textsc{Transitions}}
						\\[4pt]
\textsc{Active code} & &\textsc{Right ctx}&&\textsc{Active code} & &\textsc{Right ctx}
\\\hhline{---||-||---}

\tm\esub\var{\la\vartwo\tmtwo} & \lhd & \rctx
& \leadsto_{\symfont{sea}_1} &
\tm & \lhd & \rctxp{\ctxhole\esub\var{\la\vartwo\tmtwo}}
\\
 \tm\esub\var{\vartwo\varthree} & \lhd & \rctx
& \leadsto_{\symfont{sea}_2} &
\tm & \lhd & \rctxp{\ctxhole\esub\var{\vartwo\varthree}} & (*)
\\
\tm\esub\var{\vartwo\varthree} & \lhd & \rctx
				    & \leadsto_{\esym} &
\tm\esub{\var}{\rename{(\la\varfour\tmtwo)}\varthree} & \lhd & \rctx
& (\#)
\\
\tm\esub\var{(\la\vartwo{\evctxp{\varthree}})\varfour} & \lhd & \rctx
				    & \leadsto_{\msym} &
\evctxp{\tm\isub\var\varthree}\isub{\vartwo}{\varfour} & \lhd & \rctx
\end{array}$
\\[35pt]
(*) if $\vartwo\notin\dom\rctx$ or $\rctx(\vartwo)\neq\la\varfour\tmtwo$; \ \ \ \ \ (\#) if $\rctx(\vartwo)=\la\varfour\tmtwo$.
\end{tabular}
}
\caption{The Natural Positive Machine (Natural POM).}
\label{fig:natural-POM}
\end{figure}

%% file: 04-preliminaries_about_machines.tex
\section{Preliminaries About Abstract Machines}
\label{sect:prel-machines}
In this section, we fix the terminology about abstract machines. We follow Accattoli and co-authors \cite{DBLP:conf/lics/AccattoliC15,DBLP:conf/aplas/AccattoliBM15,DBLP:conf/ppdp/AccattoliCGC19,DBLP:journals/scp/AccattoliG19,DBLP:conf/lics/AccattoliCC21}, adapting their notions to our framework. We mostly stay abstract. In the next section, we shall instantiate the abstract notions on a specific machine.

\paragraph{Abstract Machines Glossary.}  Abstract machines manipulate \emph{pre-terms}, that is, terms without implicit $\alpha$-renaming. In this paper, an \emph{abstract 
machine} is a quadruple $\mach = (\States, \tomach, \compilrel\cdot\cdot, \decode\cdot)$ the components of which are as follows.
\begin{itemize}

\item \emph{States.} A state $\state\in\States$ is composed by the \emph{active term} $\tm$, plus some data structures; the machine of the next section shall have two data structures. Terms in states are actually pre-terms.

\item  \emph{Transitions.} The pair $(\States, \tomach)$ is a transition 
system with transitions $\tomach$ partitioned into \emph{principal transitions}, whose union is noted $\tomachpr$ and that are meant to correspond to rewriting steps on the calculus, and \emph{search transitions}, whose union is noted $\tomachsea$, that take care of searching for (principal) redexes.

\item \emph{Initialization.} The component $\compilrel{}{}\subseteq\Lambda\times\States$ is the \emph{initialization relation} associating  terms to 
initial states. It is a \emph{relation} and not a function because $\compilrel\tm\state$ maps a term $\tm$ (considered modulo $\alpha$) to a state $\state$ having a \emph{pre-term representant} of $\tm$ (which is not modulo $\alpha$) as active term. Intuitively, any two states $\state$ and $\statetwo$ such that $\compilrel\tm\state$ and $\compilrel\tm\statetwo$ are $\alpha$-equivalent. 
A state $\state$ is \emph{reachable} if it can be reached starting from an initial state, that is, if $\statetwo \tomach^*\state$ where $\compilrel\tm\statetwo$ for some $\tm$ and $\statetwo$, shortened as $\compilrel\tm\statetwo \tomach^*\state$.

\item \emph{Read-back.} The read-back function $\decode\cdot:\States\to\Lambda$ turns reachable states into 
terms and satisfies the \emph{initialization constraint}: if $\compilrel\tm\state$ then $\decode{\state}=_\alpha\tm$.
\end{itemize}

\paragraph{Further Terminology and Notations.} A state is \emph{final} if no transitions apply.
 A \emph{run} $\run: \state \tomach^*\statetwo$ is a possibly empty finite sequence of transitions, the length of which is noted 
$\size\run$; note that the first and the last states of a run are not necessarily initial and final. 
If $a$ and $b$ are transitions labels (that is, $\tomachhole{a}\subseteq \tomach$ and 
$\tomachhole{b}\subseteq \tomach$) then $\tomachhole{a,b} \defeq \tomachhole{a}\cup \tomachhole{b}$ and $\sizep\run a$ 
is the number of $a$ transitions in $\run$.

\paragraph{Well-Boundness and Renamings.} For the machine in this paper, the pre-terms in initial states shall be \emph{well-bound}, that is, they have pairwise distinct bound names; for instance $\varfour\esub\varfour{\la\varthree\varthree}\esub\var{\la\vartwo\vartwo}$ is well-bound while $\varfour\esub\varfour{\la\vartwo\vartwo}\esub\var{\la\vartwo\vartwo}$ is not. 
We shall also write $\rename{\tm}$ in a state $\state$ for a \emph{fresh well-bound renaming} of $\tm$,
\ie $\rename{\tm}$ is $\alpha$-equivalent to $\tm$, well-bound, and its bound variables
are fresh with respect to those in $\tm$ and in the other components of $\state$.

\paragraph{Mechanical Bismulations.} Machines are usually showed to be correct with respect to a strategy via some form of bisimulation relating terms and states. The notion that we adopt is here dubbed \emph{mechanical bisimulation}. The definition, tuned towards complexity analyses, requires a perfect match between the steps of the evaluation sequence and the  principal transitions of the machine run. 

\begin{definition}[Mechanical bisimulation]
A machine $\mach=(\States, \tomach, \compilrel\cdot\cdot, \decode\cdot)$ and a strategy $\tostrat$ on terms are mechanical bisimilar when, given an initial state $\compilrel\tm\state$:\label{def:implem}
\begin{enumerate}
\item \emph{Runs to evaluations}: for any run $\run: \compilrel\tm\state \tomach^* \statetwo$ there exists an evaluation $\deriv: \tm \tostrat^* \decode\statetwo$;

\item \emph{Evaluations to runs}: for every evaluation $\deriv: \tm \tostrat^* \tmtwo$ there exists a 
run $\run: \compilrel\tm\state \tomach^* \statetwo$ such that $\decode\statetwo = \tmtwo$;

\item \emph{Principal matching}: for every principal transition $\tomachhole{\symfont{a}}$ of label $a$ of $\mach$, in both previous points the number $\sizep\run{\symfont{a}}$ of $\symfont{a}$-transitions in $\run$ is exactly the number $\sizep\deriv{a}$ of of $\symfont{a}$-steps in the evaluation $\deriv$, \ie $\sizep\deriv{\symfont{a}} = \sizep\run{\symfont{a}}$.
\end{enumerate}
\end{definition}

The proof that a machine and a strategy are in a mechanical bisimulation follows from some basic properties, grouped under the notion of \emph{distillery}, following Accattoli et al. \cite{DBLP:conf/icfp/AccattoliBM14} (but removing their use of structural equivalence, that here is not needed). The intuition behind the notion of distillery is that the calculus \emph{distils} the machine in that it removes the search mechanism---as captured by the search transparency property below---as well as the organization of states in data structures, while it faithfully mimics the underlying dynamics, as captured by the principal projection and halt properties. At the meta-level, the technique is also meant to distil the reasoning behind bisimulation proofs: the mechanical bisimulation theorem after the definition is proved abstractly by relying solely on the properties defining a distillery. 
\begin{definition}[Distillery]
  A machine $\mach=(\States, \tomach, \compilrel\cdot\cdot, \decode\cdot)$ and a strategy $\tostrat$ are a \emph{distillery} if the following conditions hold:\label{def:distillery}
  \begin{enumerate}
		\item\label{p:def-beta-projection} \emph{Principal projection}: $\state \tomachhole{\symfont{a}} \statetwo$ implies $\decode\state \Rew{\symfont{a}} \decode\statetwo$ for every principal transition of label $\symfont{a}$;
		\item\label{p:def-overhead-transparency} \emph{Search transparency}: $\state \tomachsea \statetwo$ implies $\decode\state = \decode\statetwo$;
		\item\label{p:def-overhead-terminate}	\emph{Search transitions terminate}:  $\tomachsea$ terminates;

	\item\label{p:def-determinism} \emph{Determinism}: $\tostrat$ is deterministic;

	\item\label{p:def-progress} \emph{Halt}: $\mach$ final states decode to $\tostrat$-normal terms.
  \end{enumerate}
\end{definition}

\begin{toappendix}
\begin{theorem}[Sufficient condition for mechanical bisimulations]
\NoteProof{thm:abs-impl}
  Let a machine $\mach$ and a strategy $\tostrat$ be a distillery.  Then, they are mechanical bisimilar.\label{thm:abs-impl}
\end{theorem}
\end{toappendix}

%% file: 05-sliced-machine.tex
\section{The Sliced Positive Machine}
\label{sect:machine}
\input{\figurespath/figure-Sliced-POM}
In this section, we present the optimized machine for the positive $\l$-calculus and establish the bisimulation between the new machine and the right strategy $\tor$.

\paragraph{The Sliced POM: Data Structures.} The Sliced POsitive Machine (Sliced POM) is defined in \reffig{sliced-machine}. The machine has two data structures, the environment $\genv$ and the slice stack $\slice$. The active code is re-dubbed \emph{active slice}.

Environments $\genv$ are nothing else but encodings of right contexts $\rctx$. We prefer to change the representation for two reasons. Firstly, it is closer to how contexts are represented in the OCaml implementation. Secondly, it is a bit more precise: environments are free structures, and an invariant shall prove that they decode to right contexts (which are defined via some additional conditions).

The slice stack is simply a list of slices $\sliceentry\tm\var$. Note the duality between environments and slice stacks: the entries of both are pairs of a term $\tm$ and a variable $\var$, but environment entries are meant to substitute $\tm$ for $\var$ on some other term, while slices are meant to receive a variable and  substitute it for $\var$ in $\tm$.

\paragraph{The Sliced POM: New Transition.} The Sliced POM inherits the same four transitions of the Natural POM, but note that in transition $\tomachm$ now it is no longer necessary to decompose the body $\tmtwo$ of the abstraction as $\evctxp\varthree$ (please note that $\tmtwo$ is a meta-variable for \emph{terms}, and not for variables). Additionally, the Sliced POM has a new search transition $\seasym_3$. When the evaluation of the active slice is over, which happens when one is left only with the head variable $\varthree$ of the slice, the new transitions $\seasym_3$ pops the first slice on the slice stack and replaces $\varthree$ for $\var$ in $\tm$. An invariant shall guarantee that $\tm$ is a sub-term of the initial term, so that the cost of the renaming now is under control.

\paragraph{A Technical Point.} The attentive reader might wonder why $\tomachm$ and $\tomachseathree$ are not reformulated in the following eager way (with slices reduced to be simply terms), where the head of the abstraction is substituted when the slice is \emph{pushed} on the slice stack, and not when it is popped:
\begin{center}
$\begin{array}{l|l|l || c  || l|l|l |llll}
\textsc{Sl. st.} &\textsc{Active slice} &\textsc{Env}&&\textsc{Sl. st.} &\textsc{Active slice} &\textsc{Env}
\\\hline
\slice  & \tm\esub\var{(\la\vartwo\evctxp\varthree)\varfour}  & \genv
				    & \tomachm &
\slice\cons\tm\isub\var\varthree  & \evctxp\varthree\isub{\vartwo}{\varfour}  & \genv
\\[2pt]
\slice\cons\tm  & \varthree  & \genv
	& \tomachseathree &
\slice  & \tm  & \genv
\end{array}$
\end{center}
We do not adopt the eager approach because the evaluation of the active slice might change its head variable (thus making it unsound to substitute eagerly), as it is demonstrated by the following run of the standard Sliced POM:
\begin{center}
$\begin{array}{l|l|l  l}
\textsc{Slice stack} &\textsc{Active slice} &\textsc{Env}&
\\\hhline{|-|-|-|}
\emptylist& \tm\esub\var{(\la\vartwo\varthree\esub\varthree{(\la{\var'}{\var'})\varfour'})\varfour}&\emptylist & \tomachm
\\\hhline{|-|-|-|}
\emptylist\cons\sliceentry\tm\var& \varthree\esub\varthree{(\la{\var'}{\var'})\varfour'}&\emptylist & \tomachm
\\\hhline{|-|-|-|}
\emptylist\cons\sliceentry\tm\var\cons\sliceentry\varthree{\varthree}& \varfour'&\emptylist & \tomachseathree
\\\hhline{|-|-|-|}
\emptylist\cons\sliceentry\tm\var& \varfour'&\emptylist &
\end{array}$
\end{center}

%% file: figures/figure-Sliced-POM.tex
\begin{figure}[t!]
\centering
\fbox{
\begin{tabular}{cc}
$\begin{array}{c\colspace\colspace c\colspace\colspace c\colspace\colspace  c}
\textsc{Slice stacks}
&
\textsc{Environments}
&
\textsc{States}
&
\textsc{Initialization}
\\
\begin{array}{rcl}
\slice & \grameq & \emslice \mid \slice\cons \sliceentry\tm\var
\end{array}
&
\begin{array}{rcl}
\genv & \grameq & \emptylist \mid \esub\var\bite\cons \genv
\end{array}
&
\begin{array}{rcl}
\state & \grameq & \threestate\slice\tm\genv
\end{array}
&
\begin{array}{rcl}
\tm & \lhd & \threestate\emslice {\rename\tm}  \emptylist \ \ \ (*)
\end{array}

\end{array}$
			\\[10pt]
			\hline
			\textsc{Transitions}
			\\
						\footnotesize$\begin{array}{l|l|l || c  || l|l| |llll}
\textsc{Sl. st.} &\textsc{Active slice} &\textsc{Env}&&\textsc{Sl. st.} &\textsc{Active slice} &\textsc{Env}
\\\hhline{-|-|-||-||-|-|-|}

\slice  & \tm\esub\var{\la\vartwo\tmtwo}  & \genv
& \tomachseaone &
\slice  & \tm  & \esub\var{\la\vartwo\tmtwo}\cons\genv
\\[2pt]
\slice  & \tm\esub\var{\vartwo\varthree}  & \genv
& \tomachseatwo &
\slice  & \tm  & \esub\var{\vartwo\varthree}\cons\genv & (\%)
\\[2pt]
\slice  & \tm\esub\var{\vartwo\varthree}  & \genv
				    & \tomache &
\slice  & \tm\esub{\var}{\rename{(\la\varfour\tmtwo)}\varthree}  & \genv & (\#) (*)
\\[2pt]
\slice  & \tm\esub\var{(\la\vartwo\tmtwo)\varfour}  & \genv
				    & \tomachm &
\slice\cons\sliceentry\tm\var  & \tmtwo\isub{\vartwo}{\varfour}  & \genv
\\[2pt]
\slice\cons\sliceentry\tm\var  & \varthree  & \genv
	& \tomachseathree &
\slice  & \tm\isub\var\varthree  & \genv

					\end{array}$
\\[38pt]
(*) $\rename\tm$ is any well-bound code $\alpha$-equivalent to $\tm$ such that its bound\\   names are fresh
with respect to those in the rest of the state;
\\[4pt]
(\%) if $\vartwo\notin\dom\genv$ or $\genv(\vartwo)\neq\la\varfour\tmtwo$; \ \ \ \ \ (\#) if $\genv(\vartwo)=\la\varfour\tmtwo$.
\\[3pt]
\hline
$\begin{array}{r\colspace \colspace rcl\colspace\colspace rcl}
\multicolumn{7}{c}{\textsc{Read back}}
\\[2pt]
\textsc{Envs (to ctxs)} & \decode{\emptylist} & \defeq & \ctxhole
&
\decode{\esub\var\bite\cons\genv} & \defeq & \decodep\genv{\ctxhole\esub\var\bite}
\\[3pt]
\textsc{States (to terms)} & \decode{\threestate\emslice\tm\genv} & \defeq & \decodep\genv\tm
&
\decode{\threestate{\slice\cons\sliceentry\tm\var}{\evctxp\vartwo}\genv} & \defeq & \decode{\threestate\slice {\evctxp{ \tm\isub\var\vartwo }} \genv}
\end{array}$

\end{tabular}

}
\caption{The Sliced Positive Machine (Sliced POM).}
\label{fig:sliced-machine}
\end{figure}

%% file: 06-mechanical_bisimulation.tex

\paragraph{Invariants.} To establish the  bisimulation between the strategy $\tor$ of the positive $\l$-calculus and the Sliced POM, we shall prove that the two form a distillery, that, in turn, is proved using the following invariants of the machine.
\label{sect:bisimulation}
\begin{lemma}[Qualitative invariants]
Let $\state = \threestate\slice\tm\genv$ be a reachable state. \label{l:invariants}
\begin{enumerate}
  \item \label{p:invariants-rb} \emph{Contextual read-back}:
  $\decode\genv$ is a right context.

  \item \label{p:invariants-bv} \emph{Well-bound}: 
  \begin{itemize}
  \item \emph{Bound names in terms}: if $\la\var\tmtwo$ or $\tmtwo\esub\var\bite$ occurs in $\state$ then any
other occurrence of $\var$ in $\state$, if any, is a free variable occurrence of $\tmtwo$;
  \item \emph{ES names of environment entries}: for any ES $\esub\vartwo\bite$ in $\genv$ the name $\vartwo$ can occur (in any form) only on the left of that ES in $\state$.
  \end{itemize}
\end{enumerate}
\end{lemma}
\begin{proof}
By induction on the length of the run reaching $\state$. For both points, the base case trivially holds, and the inductive case is by analysis of the last transition, which is always a straightforward inspection of the transitions using the \ih For contextual read-back, the proof relies on the inside-out definition of right contexts (see \refdef{right-ctx}).\qed
\end{proof}
The well-bound invariant has two consequences that might not be evident. Firstly, there cannot be two ESs $\esub\var\bite$ and $\esub\var\bitetwo$ on the same variable. This fact implies the determinism of the machine, that however shall not be proved because it is not necessary (determinism of the strategy is enough for proving the bisimulation). Secondly, the name $\var$ of an environment entry $\esub\var\bite$ can occur both in the active slice and in the slice stack, while the names of ESs in slices can occur only within the slice. In particular, $\var$ cannot occur in $\slice$ in $\threestate\slice{\tm\esub\var\bite}\genv$.

\begin{toappendix}
\begin{theorem}[Distillery]
\label{thm:distillery}\NoteProof{thm:distillery}\hfill
\begin{enumerate}
  \item Principal projection: 
    \begin{enumerate}
      \item if $\state \tomache \statetwo$ then $\decode{\state} \toe \decode{\statetwo}$.
      \item if $\state \tomachm \statetwo$ then $\decode{\state} \tom \decode{\statetwo}$.
    \end{enumerate}
  \item Search transparency: if $\state \leadsto_{\symfont{sea}_1, \symfont{sea}_2, \symfont{sea}_3} \statetwo$ then $\decode{\state} = \decode{\statetwo}$.
  \item Search terminates: $\leadsto_{\symfont{sea}_1, \symfont{sea}_2, \symfont{sea}_3}$ is terminating.
  \item Halt: if $\state$ is final then it is of the form $\threestate\emslice\var\genv$ and $\decode{\state} = \decodep\genv\var$ is $\topos$-normal.
\end{enumerate}
\end{theorem}
\end{toappendix}

The points of the proved theorem together with determinism of the right strategy (\reflemma{right-strat-properties}) provide all the requirements for a distillery. Then, the abstract theorem about distilleries (\refthm{abs-impl}) gives the following corollary.

\begin{corollary}[Mechanical bisimulation]
The Sliced POM and the right strategy $\tor$ are in a mechanical bisimulation.
\end{corollary}

%% file: 07-complexity.tex
\section{Complexity Analysis}
\label{sect:complexity}
In this section, we show that the Sliced POM can be concretely implemented within a bi-linear overhead, that is, linear in the number of $\msym$-steps/transitions and the size of the initial term.

\paragraph{Sub-Term Property.} As it is standard, the complexity analysis crucially relies on the sub-term property. In \cbv settings, the property is usually expressed saying that all values are sub-terms of the initial terms, as in \reflemma{sub-term-calculus}. The Sliced POM only duplicates values too, but (as we have discussed for the Natural POM) we also want to know that the terms to which renamings are applied are sub-terms of the initial term, and these terms are not values. Thus, the property is given with respect to \emph{all} terms in a state.

There is in fact a very minor exception. The substitution performed by an exponential transition takes two sub-terms of the initial term, namely $\tm\esub{\var}{\vartwo\varthree}$ and $\la\varfour\tmtwo$, and creates a term $\tm\esub{\var}{\rename{(\la\varfour\tmtwo)}\varthree}$ that is not a sub-term of the initial one. The new term, however, is very short lived: the next transition is multiplicative and decomposes that term in sub-terms of the initial term (up to renaming).

\begin{toappendix}
\begin{lemma}[Sub-term property]
Let $\compilrel\tm\state \tomach^*\statetwo=\threestate\slice\tmtwo\genv$ be a Sliced POM run. Then:\label{l:sub-term}
\begin{enumerate}
\item If $\statetwo$ is not the target of a $\esym$-transition then $\size\tmthree\leq\size\tm$ for any term $\tmthree$ in $\statetwo$;
\item Otherwise, $\size\tmthree\leq\size\tm$ for any term $\tmthree$ in $\statetwo$ except $\tmtwo$.
\end{enumerate}
\end{lemma}
\end{toappendix}
\begin{proof}
By induction on the length of the run, inspecting the last transition and using the \ih\qed
\end{proof}

\paragraph{Number of Transitions.} Some basic observations about the transitions, together with the sub-term property, allow us to bound their number using the two key parameters (that is, number of $\msym$-steps/transitions and size of the initial term).

\begin{lemma}[Number of transitions]
Let $\run:\compilrel\tm\state \tomach^*\statetwo$ be a Sliced POM run.$\label{l:number-single-transitions}$
\begin{enumerate}
\item $\sizep\run{\esym,\seasym_3}\in\bigo(\sizep\run\msym)$;
\item $\sizep\run{\seasym_1,\seasym_2}\in\bigo(\size{\tm}\cdot(\sizep\run\msym +1))$.
\end{enumerate}
\end{lemma}

\begin{proof}

\begin{enumerate}
\item $\sizep\run{\esym} \le \sizep\run\msym+1$ because exponential transitions can only be followed by multiplicative transitions.

$\sizep\run{\seasym_3} \le \sizep\run\msym$ because every $\seasym_3$ transition consumes one entry from the slice stack, which are created only by $\msym$ transitions. 
\item Note that $\seasym_1/\seasym_2$ transitions decrease the size of the active slice, which is increased only by transitions $\esym$ and $\seasym_3$. By the sub-term property (\reflemma{sub-term}), the size increase of the active slice by transitions $\esym$ and $\seasym_3$ is bounded by the size $\size\tm$ of the initial term. By Point 1, $\sizep\run{\esym,\seasym_3} =\bigo(\sizep\run\msym)$. Then $\sizep\run{\seasym_1,\seasym_2} \in \bigo(\size{\tm}\cdot(\sizep\run{\esym,\seasym_2} +1))=\bigo(\size{\tm}\cdot(\sizep\run\msym +1))$.\qed
\end{enumerate}
\end{proof}

\paragraph{Cost of Single Transitions.} Lastly, we need some assumptions on how the Sliced POM can be concretely implemented. Transition $\seasym_2$ can evidently be done in $\bigo(1)$. Our OCaml implementation---overviewed in \withproofs{\refapp{app-implementation}}\withoutproofs{Appendix A of the technical report \cite{techreport}}---represents variables as memory locations and variable occurrences as pointers to those locations, obtaining random access to environment entries in $\bigo(1)$.\footnote{Assuming that memory accesses take $\bigo(1)$ is an idealized abstraction. In real computer architectures, that cost is highly dependent on where the data is stored, but it is bounded nonetheless.
To be theoretically precise, the cost of accessing an \emph{unbound} memory should be instead assumed to be logarithmic in the size of the memory in use. In the analyses of polynomial algorithms, and of abstract machines in particular, it is standard to assume an underlying random access machine model with constant-time access, since the logarithmic factor is somewhat negligible. The index of the location does instead play a more relevant role in the study of lower time and space complexities.} Therefore, also transition $\seasym_1$ can be done in $\bigo(1)$. The cost of transition $\esym$ is bound by the size of the value to copy, itself bound by the sub-term property. The cost of transitions $\msym$ and $\seasym_3$ is bound by the size of the term to rename, itself bound by the sub-term property.  The next lemma sums it up.

\begin{lemma}[Cost of single transitions]
  Let $\compilrel\tm\state \tomach^*\statetwo$ be a Sliced POM run. Implementing $\seasym_1$ and $\seasym_2$ transitions from $\statetwo$ costs $\bigo(1)$ each while implementing $\esym$, $\msym$, and $\seasym_3$ transitions costs $\bigo(\size\tm)$ each.\label{l:cost-single-transitions}
\end{lemma}

Putting all together, we obtain our main result: a bilinear bound for the Sliced POM, showing that it is an efficient machine for the right strategy.
\begin{theorem}[Sliced POM is bi-linear]
Let $\run: \compilrel\tm\state \tomach^*\statetwo$ be a Sliced POM run. Then $\run$ can be implemented on random access machines in $\bigo(\size{\tm}\cdot(\sizep\run\msym +1))$.$\label{thm:SlicedPOM-is-bilinear}$
\end{theorem}

\begin{proof}
The cost of implementing $\run$ is obtained by multiplying the number of each kind of transitions (\reflemma{number-single-transitions}) by the cost of that kind of transition (\reflemma{cost-single-transitions}), and summing over all kinds of transition.\qed
\end{proof}

%% file: 99-conclusions.tex
\section{Conclusions}
\label{sect:conclusions}
In $\l$-calculi with sharing, renaming chains are a recurrent issue that causes both time and space inefficiencies. The recently introduced positive $\l$-calculus removes renaming chains, while adding some meta-level renamings. This paper stems from the observation that the added meta-level renamings reintroduce a time inefficiency if implemented naively.

The problem is analyzed via the sub-term property, showing that the culprit is the fact that the scope of renamings is not a sub-term of the initial term. The analysis leads to the design of an optimized machine, the new Sliced POM, that removes once and for all the inefficiency. The key tool is a new decomposition in slices of the scopes of renamings, via a new stack for slices playing a role dual to that of environments. We also provide a prototype OCaml implementation of the Sliced POM, described in \withproofs{\refapp{app-implementation}}\withoutproofs{Appendix A of the technical report \cite{techreport}}.

\paragraph{Future Work.} At the theoretical level, we plan to adapt the positive $\l$-calculus and the Sliced POM to call-by-need evaluation. At the practical level, it would be interesting to see how the schema of the Sliced POM combines with other techniques such as closure conversion or skeletal call-by-need; these techniques were in fact recasted in the same abstract machine framework of the present work, and also analyzed from a complexity point of view, in two parallel works involving Accattoli and Sacerdoti Coen \cite{DBLP:journals/corr/abs-2507-15843,DBLP:conf/fscd/AccattoliMPC25}. We also suspect that the Sliced POM can be used to simplify the sophisticated machine for Strong Call-by-Value by Accattoli et al. in \cite{DBLP:conf/lics/AccattoliCC21}.

%% file: APP-00-implementation.tex

\definecolor{light-gray}{gray}{0.95}

\lstset{
         language=[Objective]{Caml},
         basicstyle=\small,
         backgroundcolor=\color{light-gray},
         columns=fullflexible,
         morekeywords=[2]{Some,None,option,list,int},
         keywordstyle={[2]\color{teal}},
         identifierstyle=\color{blue},
         commentstyle=\color{orange},
         stringstyle=\color{gray}\emph,
         mathescape=true
}

\section{Implementation in OCaml}
\label{sect:app-implementation}

The implementation is published on GitHub: \url{https://github.com/sacerdot/PositiveAbstractMachine}

\paragraph{Compiling and Running the Prototype.}
The \texttt{README} explains how to compile and run the prototype. The code requires a working \texttt{dune} installation and it depends on the \texttt{js\_of\_ocaml} package, since it provides two different interfaces: a Read-Eval-Print-Loop (REPL) from the command line or a Web interface, i.e. an HTML+JavaScript page that performs all the computation on the client side. The two interfaces are functionally equivalent. The concrete syntax for $\lambda$-terms that is accepted by the implementation is also explained in the \texttt{README} file. A screenshot of the Web interface is shown in Table~\ref{tbl-webint} where the
$\lambda$-term $(\lambda x.xx)((\lambda z.z)(\lambda z.z))$ is transformed into a positive $\l$-term (via the transformation studied by Accattoli and Wu in~\cite{entics:14758}), namely (the variable names generated by the implementation are all numbered variants of $v$):
\begin{center}
$v_4\esub{v_4}{(\underbrace{\la{v_1}v_5\esub{v_5}{v_1v_1}}_{\lambda x.xx})v_6}\esub{v_6}{(\underbrace{\la{v_2}v_2}_{\lambda z.z})v_7}\esub{v_7}{\underbrace{\la{v_3}v_3}_{\lambda z.z}}$
\end{center}
and then executed using the Sliced POM.

\begin{table}[t]
\includegraphics[width=\textwidth]{webinterface-new.png}
  \caption{A screenshot from the Web interface. The {\color{blue} active slice} is highlighted.\label{tbl-webint}}
\end{table}

\paragraph{Code Structure.} The implementation consists of five OCaml files. The first four, \texttt{utils.ml}, \texttt{lc.ml}, \texttt{repl.ml} and \texttt{main.ml}, are not interesting: they provide respectively utility functions, e.g. to output nicely both text and HTML, a parser and pretty-printer for $\lambda$-calculus terms, the command line REPL loop, and its equivalent for the Web interface. The final one,
\texttt{spom.ml}, is the intersting one. It declares the data types for positive terms and Sliced POM states, pretty-printing functions for them, the transformation from $\lambda$-terms to positive terms (also called \emph{crumbling}, following Accattoli et al. \cite{DBLP:conf/ppdp/AccattoliCGC19}), functions to immediately substitute a variable for another in a term, $\alpha$-renaming functions (i.e. copy), and the main loop of the Sliced POM.

\paragraph{Data Structures for Positive Terms.} Positive terms are represented by the following data structure.

\begin{lstlisting}
type term =
 | V of var
 | ESubst of term * var
and var =
 { name : int
 ; mutable subst : bite option
 ; mutable copying_to : var option
 }
and bite =
 | VV of var * var
 | AV of abst * var
 | A of abst
and abst = var * term
\end{lstlisting}

Variables are uniquely represented in memory by a single record \lstinline{var} that is referenced by every
variable occurrence (i.e. the reference $v$ occurs in \lstinline{V $v$}, \lstinline{VV($v$,$v'$)}, \lstinline{VV($v'$,$v$)}, \lstinline{AV($a$,$v$)}). When a variable is bound by an explicit
substitution, its field \lstinline{subst} is set to \lstinline{Some $b$} where $b$ is the definiens. This allows to retrieve in $\bigo(1)$ the definiens of a variable given a variable occurrence.
The term $\tm\esub\var\tmtwo$ is represented by \lstinline{ESubst($\tm$,$\var$)} since the record referenced by $\var$ contains $\tmtwo$ in its \lstinline{subst} field.

We use integer numbers for variable names to make it simpler to generate fresh names. A variable such that \lstinline{name=$n$} will be pretty-printed as $v_n$.

The \lstinline{copying_to} field is set only during $\alpha$-renaming (i.e. sharing preserving copy of terms): when a binder is encountered during copy for the first time, a new fresh variable is
generated and the original bound variable is made to point to the new fresh variable using \lstinline{copying_to}. When later an occurrence of the bound variable is found during the copy, the new variable
reference is obtained from \lstinline{copying_to}, preserving sharing.

\paragraph{Data Structures for Sliced POM States.} Machine states are represented by the following data structures, that follow closely the definitions in the paper.

\begin{lstlisting}
type env = var list
type slice = (term * var) list
type state = slice * term * env
\end{lstlisting}

Previous work on crumbled terms by Accattoli et al.~\cite{DBLP:conf/ppdp/AccattoliCGC19,DBLP:conf/lics/AccattoliCC21} avoids the OCaml \lstinline{list} type and implements environments (and other list-based data structures) by adding to variable nodes an
optional field \lstinline{next} pointing to the next entry in the list. This allows one to see the pairs \lstinline{term-env} simply as a zipper data structure where variables in the term point to
the next inner substitution and variables in the environment point to the next outer substitution. Sometimes an additional \lstinline{prev} pointer is also added, turning the lists into bidirectional lists to
allow for constant time appending, which is required for some reduction machines.

In the implementation for this paper, we preferred to use standard OCaml \lstinline{list}s to keep the code closer to the paper. Switching to the other representation would not change the computational complexity, but would
diminish the constant for space usage.

\paragraph{Main Sliced POM Loop.} We show here an excerpt of the code of the Sliced POM main loop, that implements the rules in the paper practically verbatim.

\begin{lstlisting}
let rec eval ?last_rule_used (state:state) =
 ...
 match state with
   ...
  | (sl,ESubst(t,({subst=Some (A _);_} as subst)),env) ->
     (* sea$\color{orange}{_1}$ *)
     eval ~last_rule_used:"sea$_1$" (sl,t,subst::env)
  | (_,ESubst(_,({subst=Some(VV({subst=Some (A a);_},z));_} as var)),_)
       as state ->
     (* e *)
     var.subst <- Some(AV(alpha_abs a,z)) ;
     eval ~last_rule_used:(Utils.pad 3 "e") state
  | (sl,ESubst(t,({subst=Some(AV((y,ez),w));_} as x)),env) ->
     (* m *)
     x.subst <- None ;
     eval ~last_rule_used:(Utils.pad 3 "m")
      ((t,x)::sl,replace ~what:y ~with_:w ez,env)
  ...
\end{lstlisting}

You can see that the \lstinline{alpha_abs} function is used in the exponential rule to rename an abstraction and
that \lstinline{replace} is used in the multiplicative rule to immediately substitue the variable $y$ with the
variable $w$. Those functions are simply defined by recursion on the syntax of positive terms as one would expect.

Note that all recursive calls are in tail position, making the recursive function equivalent to a simple loop.
The optional \lstinline{last_rule_used} parameter is used only for pretty-printing purposes, to show each machine transition to the user.

%% file: APP-01-proofs.tex
\section{Proofs Omitted from \refsect{right-strategy} (The Right Strategy)}

\begin{lemma}[Outer extension of inside-out is inside-out]
Let $\irctx$ be an inside-out right context. Then:\label{l:outer-extension}
\begin{enumerate}
\item $\irctx\esub\var{\vartwo\varthree}$ is an inside-out right context;\label{p:outer-extension-one}
\item If $\var\notin\afv\irctx$ then $\irctx\esub\var{\la\vartwo\tmtwo}$ is an inside-out right context.\label{p:outer-extension-two}
\end{enumerate}
\end{lemma}

\begin{proof}
\hfill
\begin{enumerate}
\item By induction on $\irctx$. The base case is obvious.
\begin{itemize} 
\item $\irctx = \irctxtwop{\ctxhole\esub{\var'}{\vartwo'\varthree'}}$ with $\irctxtwo(\vartwo')\neq \la\varfour\tm$. By \ih, $\irctxtwo\esub\var{\vartwo\varthree}$ is an inside-out right context. Then $\irctxtwop{\ctxhole\esub{\var'}{\vartwo'\varthree'}}\esub\var{\vartwo\varthree}$ is an inside-out right context.

\item $\irctx = \irctxtwop{\ctxhole\esub{\var'}{\la{\vartwo'}\tmtwo}}$. By \ih, $\irctxtwo\esub\var{\vartwo\varthree}$ is an inside-out right context. Then $\irctxtwop{\ctxhole\esub{\var'}{\la{\vartwo'}\tmtwo}}\esub\var{\vartwo\varthree}$ is an inside-out right context.
\end{itemize}

\item By induction on $\irctx$. The base case is obvious.
\begin{itemize} 
\item $\irctx = \irctxtwop{\ctxhole\esub{\var'}{\vartwo'\varthree'}}$ with $\irctxtwo(\vartwo)\neq \la\varfour\tm$. Since $\var\notin\afv\irctxtwo$, by \ih we obtain that $\irctxtwo\esub\var{\la\vartwo\tmtwo}$  is an inside-out right context. By hypothesis, $\irctxtwo(\vartwo)\neq \la\varfour\tm$, and the hypothesis $\var\notin\afv\irctx$ implies that $\vartwo'\neq\var$, so that $(\irctxtwo\esub\var{\la\vartwo\tmtwo})(\vartwo)\neq \la\varfour\tm$. Then $\irctxtwop{\ctxhole\esub{\var'}{\vartwo'\varthree'}}\esub\var{\la\vartwo\tmtwo}$  is an inside-out right context.

\item $\irctx = \irctxtwop{\ctxhole\esub{\var'}{\la{\vartwo'}\tmtwo}}$. Since $\var\notin\afv\irctxtwo$, by \ih we obtain that $\irctxtwo\esub\var{\la\vartwo\tmtwo}$  is an inside-out right context. Then $\irctxtwop{\ctxhole\esub{\var'}{\la{\vartwo'}\tmtwo}}\esub\var{\la\vartwo\tmtwo}$ is an inside-out right context.\qed
\end{itemize}
\end{enumerate}
\end{proof}

\begin{lemma}[Inner extension of outside-in is outside-in]
Let $\orctx$ be an outside-in right context. Then:\label{l:inner-extension}
\begin{enumerate}
\item If $\orctx(\vartwo) \neq \la\varfour\tm$ then $\orctxp{\ctxhole\esub\var{\vartwo\varthree}}$ is an outside-in right context;\label{p:inner-extension-one}
\item $\orctxp{\ctxhole\esub\var{\la\vartwo\tmtwo}}$ is an outside-in right context.\label{p:inner-extension-two}
\end{enumerate}
\end{lemma}

\begin{proof}
\hfill
\begin{enumerate}
\item By induction on $\orctx$. The base case is obvious.
\begin{itemize}
\item $\orctx = \orctxtwo\esub{\var'}{\vartwo'\varthree'}$. By \ih, $\orctxtwop{\ctxhole\esub\var{\vartwo\varthree}}$ is an outside-in right context. Then $\orctxtwop{\ctxhole\esub\var{\vartwo\varthree}}\esub{\var'}{\vartwo'\varthree'}$ is an outside-in right context.
\item $\orctx = \orctxtwo\esub{\var'}{\la{\vartwo'}\tmtwo}$ with $\var' \notin \afv\orctxtwo$. By \ih, $\orctxtwop{\ctxhole\esub\var{\vartwo\varthree}}$ is an outside-in right context. By hypothesis, $\orctx(\vartwo) \neq \la\varfour\tm$, which implies $\var' \neq \vartwo$. Then $\orctxtwop{\ctxhole\esub\var{\vartwo\varthree}}\esub{\var'}{\la{\vartwo'}\tmtwo}$ is an outside-in right context.
\end{itemize}
\item By induction on $\orctx$. The base case is obvious.
\begin{itemize}
\item $\orctx = \orctxtwo\esub{\var'}{\vartwo'\varthree'}$. By \ih, $\orctxtwop{\ctxhole\esub\var{\la\vartwo\tmtwo}}$ is an ouside-in right context. Then $\orctxtwop{\ctxhole\esub\var{\la\vartwo\tmtwo}}\esub{\var'}{\vartwo'\varthree'}$ is an outside-in right context.
\item $\orctx = \orctxtwo\esub{\var'}{\la{\vartwo'}\tmtwo}$ with $\var' \notin \afv\orctxtwo$. By \ih, $\orctxtwop{\ctxhole\esub\var{\la\vartwo\tmtwo}}$ is an ouside-in right context. Then $\orctxtwop{\ctxhole\esub\var{\la\vartwo\tmtwo}}\esub{\var'}{\la{\vartwo'}\tmtwo}$ is an outside-in right context.\qed
\end{itemize}
\end{enumerate}
\end{proof}

\gettoappendix{l:right-ctxs-coincide}

\begin{proof}
\applabel{l:right-ctxs-coincide}\hfill
\begin{itemize}
\item \emph{Outside-in are inside-out}. Let $\orctx$ be an outside-in right context. By induction on $\orctx$. The base case is obvious.
\begin{itemize}
\item $\orctx = \orctxtwo\esub\var{\vartwo\varthree}$. By \ih, $\orctxtwo$ is an inside-out right context. By \reflemmap{outer-extension}{one}, $\orctxtwo\esub\var{\vartwo\varthree}$ is an inside-out right context.
\item $\orctx = \orctxtwo\esub\var{\la\vartwo\tmtwo}$ with $\var\notin\afv\orctx$. By \ih, $\orctxtwo$ is an inside-out right context. By \reflemmap{outer-extension}{two}, $\orctxtwo\esub\var{\la\vartwo\tmtwo}$ is an inside-out right context.
\end{itemize}

\item \emph{Inside-out are outside-in}. Let $\irctx$ be an inside-out right context. By induction on $\irctx$. The base case is obvious.
\begin{itemize}
\item $\irctx = \irctxtwop{\ctxhole\esub\var{\vartwo\varthree}}$ with $\irctxtwo(\vartwo) \neq \la\varfour\tm$. By \ih, $\irctxtwo$ is an outside-in right context. By \reflemmap{inner-extension}{one}, $\irctxtwop{\ctxhole\esub\var{\vartwo\varthree}}$ is an outside-in right context.
\item $\irctx = \irctxtwop{\ctxhole\esub\var{\la\vartwo\tmtwo}}$. By \ih, $\irctxtwo$ is an outside-in right context. By \reflemmap{inner-extension}{two}, $\irctxtwop{\ctxhole\esub\var{\la\vartwo\tmtwo}}$ is an outside-in right context.\qed
\end{itemize}
\end{itemize}
\end{proof}

\gettoappendix{l:right-strat-properties}
\input{proofs/proof_lemma_2}

\input{APP-Abstract_machines}

\section{Proofs Omitted from \refsect{machine} (The Sliced Positive Machine)}

\gettoappendix{thm:distillery}
\input{proofs/proof_distillery}

%% file: proofs/proof_lemma_2.tex
\begin{proof}
\applabel{l:right-strat-properties}\hfill
\begin{enumerate}
\item Let $\rctx$ and $\rctxtwo$ be the positions of the two redexes. Without loss of generality, we assume that $\rctxtwo$ strictly extends $\rctx$ inward, that is, there exists $\evctx \neq \ctxhole$ such that $\rctxtwo=\rctxp\evctx$. Then we consider the different cases for the redex of position $\rctx$ and we shall derive every time a contradiction. The only possibility left shall be that $\rctx=\rctxtwo$. Cases of the redexes:
\begin{itemize}
\item $\rctx$ is the position of a $\rsym\esym$-redex. Then $\rctxtwo=\evctxtwop{\evctxthreep{\evctxfour\esub{\var}{\vartwo\varthree}}\esub\vartwo{\la\varfour\tm}}$, which is not a right context because $\var\in\afv{\evctxthreep{\evctxfour\esub{\var}{\vartwo\varthree}}}$. Absurd.

\item $\rctx$ is the position of a $\rsym\msym$-redex. Then $\rctxtwo=\evctxtwop{\evctxthree\esub\var{(\la\vartwo\tm)\varthree}}$, which is not a right context because no right contexts contains ESs of shape $\esub\var{(\la\vartwo\tm)\varthree}$.
\end{itemize}

\item Let $\evctx$ be the position of $\tm\topos \tmtwo$. If $\evctx$ is a right context the statement holds. Otherwise, let $\rctx$ be the maximum right context that is a prefix of $\evctx$, that is, such that $\evctx=\rctxp\evctxtwo$ for some $\evctxtwo$. We show that it is the position of a redex. Let $\tmfour$ be the sub-term isolated by $\rctx$, that is, such that $\tm = \rctxp\tmfour$. Cases of $\tmfour$:
\begin{itemize}
\item $\tmfour = \tmfour'\esub\var{(\la\vartwo\tmfour'')\varthree}$. Then $\rctx$ is the position of a $\rsym\msym$-redex.
\item $\tmfour = \tmfour'\esub\var{\la\vartwo\tmfour''}$. Impossible, because then $\rctxp{\esub\var{\la\vartwo\tmfour''}}$ is a right context bigger than $\rctx$ and prefix of $\evctx$, against maximality of $\rctx$.
\item $\tmfour = \tmfour'\esub\var{\vartwo\varthree}$. By maximality of $\rctx$, one has that $\rctx(\vartwo)$ is an abstraction, so that $\rctx$ is the position of a $\rsym\esym$-redex.\qed
\end{itemize}
\end{enumerate}
\end{proof}

%% file: APP-Abstract_machines.tex
\section{Proofs omitted from \refsect{prel-machines} (Preliminaries about Abstract Machines)}
\label{sect:app-prel-machines}

\begin{lemma}[One-step simulation]
  \label{l:one-step-simulation}
  Let a machine $\mach$ and a strategy $\tostrat$ be a distillery.
  For any state $\state$ of $\mach$, if $\decode\state \tostrat\tmtwo$ then there is a state $\statetwo$ of $\mach$ such that $\state \tomachsea^*\tomachpr \statetwo$ and $\decode{\statetwo} = \tmtwo$ and such that the label of the principal transition and of the step are the same.
\end{lemma}

\begin{proof}
  For any state $\state$ of $\mach$, let $\nfo{\state}$ be a normal form of $\state$ with respect to $\tomachsea$: such a state exists because search transitions terminate (\refpoint{def-overhead-terminate} of \refdef{distillery}).
  Since $\tomachsea$ is mapped on identities (by \emph{search transparency}, \ie \refpoint{def-overhead-transparency} of \refdef{distillery}), one has $\decode{\nfo{\state}} = \decode\state$.
  Since $\decode\state$ is not $\tostrat$-normal by hypothesis, the halt property (\refpoint{def-progress}) entails that $\nfo{\state}$ is not final, therefore $\state \tomachsea^* \nfo{\state} \tomachpr \statetwo$ for some state $\statetwo$, and thus $\decode\state = \decode{\nfo{\state}} \tostrat \decode{\statetwo}$ by principal projection (\refpoint{def-beta-projection}).
  By determinism of $\tostrat$ (\refpoint{def-determinism}), one obtains $\decode{\statetwo} = \tmtwo$.\qed
\end{proof}

\gettoappendix{thm:abs-impl}
\begin{proof}
\applabel{thm:abs-impl}
  According to \refdef{implem}, given a positive term $\tm$ and an initial state $\compilrel\tm\state$, we have to show that:
  \begin{enumerate}
   \item \label{p:exec-to-deriv} \emph{Runs to evaluations}: for any $\mach$-run $\run: \compilrel\tm\state \tomachine^* \statetwo$ there exists a 
$\tostrat$-evaluation $\deriv: \tm \tostrat^* \decode\statetwo$;

\item \label{p:deriv-to-exec} \emph{Evaluations to runs}: for every $\tostrat$-evaluation $\deriv: \tm \tostrat^* \tmtwo$ there exists a 
$\mach$-run $\run: \compilrel\tm\state \tomachine^* \statetwo$ such that $\decode\statetwo = \tmtwo$;
  \end{enumerate}
  Plus the principal matching constraint that shall be evident by the principal projection requirement and how the proof is built.

  \paragraph{Proof of \refpoint{exec-to-deriv}.}  By induction on $\sizepr\run \in \nat$. Cases:
  \begin{itemize}
  \item $\sizepr\run = 0$. Then $\run \colon \compilrel\tm\state \tomachsea^* \statetwo$ and hence $\decode{\state} = \decode\statetwo$ by search transparency (\refpoint{def-overhead-transparency} of \refdef{distillery}).
  Moreover, $\tm = \decode{\state}$ since decoding is the inverse of initialization, therefore the statement holds with respect to the empty evaluation $\deriv$ with starting and end term $\tm$.
  
\item $\sizepr\run > 0$. Then, $\run \colon \compilrel\tm\state \tomachine^* \statetwo$ is the concatenation of a run $\runtwo \colon \compilrel\tm\state \tomachine^* \statethree$ followed by a run $\runthree \colon \statethree \tomachpr \statefour \tomachsea^* \statetwo$.
  By \ih applied to $\runtwo$, there exists an evaluation $\derivtwo \colon \tm \tostrat^*  \decode\statetwo$ satisfying the principal matching constraint.
  By principal projection (\refpoint{def-beta-projection} of \refdef{distillery}) and search transparency (\refpoint{def-overhead-transparency} of \refdef{distillery}) applied to $\runthree$, one obtains a one-step evaluation $\derivthree \colon \decode\statetwo \tostrat   \decode\state$ having the same principal label of the transition. 
Concatenating $\derivtwo$ and $\derivthree$, we obtain an evaluation $\derivfour \colon \tm  \tostrat^*  \decode\statetwo\tostrat  \decode\state$.
 \end{itemize}
 
  \paragraph{Proof of \refpoint{deriv-to-exec}.}  By induction on $\size\deriv \in \nat$. Cases:
  \begin{itemize}
  \item $\size\deriv = 0$. Then $\tm = \tmtwo$.
  Since decoding is the inverse of initialization, one has $\decode\state = \tm$.
  Then the statement holds with respect to the empty (\ie without transitions) run $\run$ with initial (and final) state $\state$.
  
 \item  $\size\deriv > 0$. Then, $\deriv\colon \tm \tostrat^* \tmtwo$ is the concatenation of an evaluation $\derivtwo \colon \tm \tostrat^* \tmtwop$ followed by the step $\tmtwop \tostrat \tmtwo$.
  By \ih, there exists a $\mach$-run $\runtwo\colon \compilrel\tm\state \tomachine^* \statethree$ such that $\decode\statethree  = \tmtwop$ and verifying the principal matching constraint.
  Since $\decode\statethree  = \tmtwop\tostrat \tmtwo$.
By one-step simulation (\reflemma{one-step-simulation}), there is a state $\statetwo$ of $\mach$ such that $\statethree \tomachsea^*\tomachpr \statetwo$ and $\decode\statetwo = \tmtwo$, preserving the label of the step/transition.
  Therefore, the run $\run \colon \compilrel\tm\state \tomachine^*\statethree \tomachsea^*\tomachpr \statetwo$ satisfies the statement.\qed
  \end{itemize}
\end{proof}

%% file: proofs/proof_distillery.tex
\begin{proof}
\applabel{thm:distillery}\hfill
\begin{enumerate}
  \item 
  \begin{enumerate}
    \item If $\state = \threestate\slice{\tm\esub\var{\vartwo\varthree}}{\genv} \tomache \threestate\slice{\tm\esub\var{(\la\varfour\tmtwo)\varthree}}{\genv} = \statetwo$ with $\genv(\vartwo)= \la\varfour\tmtwo$ then $\genv = \genvtwo\cons\esub\vartwo{\la\varfour\tmtwo}\cons\genvthree$ for some $\genvtwo$ and $\genvthree$. We proceed by induction on $\slice$. Cases:
    \begin{itemize}
      \item \emph{Empty slice stack}, \ie $\slice = \emslice$. We have: \[
      \begin{array}{rll}
      &&\decode{\threestate\emslice{\tm\esub\var{\vartwo\varthree}}{\genvtwo\cons\esub\vartwo{\la\varfour\tmtwo}\cons\genvthree}} 
      \\[3pt]
      & = & \decodep\genvthree{\decodep\genvtwo{\tm\esub\var{\vartwo\varthree}}\esub\vartwo{\la\varfour\tmtwo}} 
      \\[3pt]
      & \toe & \decodep\genvthree{\decodep\genvtwo{\tm\esub\var{(\la\varfour\tmtwo)\varthree}}\esub\vartwo{\la\varfour\tmtwo}}  
      \\[3pt]
      & = & \decode{\threestate\emslice{\tm\esub\var{\rename{(\la\varfour\tmtwo)}\varthree}}{\genvtwo\cons\esub\vartwo{\la\varfour\tmtwo}\cons\genvthree}}
      \end{array} \]
      Note that the exponential step applies because of:
      \begin{itemize}
      \item The well-bound invariant (\reflemmap{invariants}{bv}) that guarantees that $\genvtwo$ does not capture $\vartwo$ (because by the invariant there cannot be two ESs on the same name);
      \item The contextual read-back invariant (\reflemmap{invariants}{rb}) that guarantees that $\decode\genvtwo$ and $\decode\genvthree$ are right contexts.
      \end{itemize}
      
      \item \emph{Non-empty slice stack}, \ie $\slice = \slicetwo \cons \sliceentry\tmthree\varfive$. 
      Let $\tm = \evctxp\varsix$.
      We have: 
       \[\begin{array}{rllr}
      \hhline{~-|-|-}
            &(\slicetwo \cons \sliceentry\tmthree\varfive &, \tm\esub\var{\vartwo\varthree} & ,\genv)
      \\[3pt]\hhline{~-|-|-}
          =  &(\slicetwo \cons \sliceentry\tmthree\varfive &, \evctxp\varsix\esub\var{\vartwo\varthree} & ,\genv)
      \\[3pt]\hhline{~-|-|-}
          = &(\slicetwo &, \evctxp{\tmthree\isub\varfive\varsix}\esub\var{\vartwo\varthree} & ,\genv)
      \\[3pt]\hhline{~-|-|-}
          (\text{by \ih})  \toe  &(\slicetwo &, \evctxp{\tmthree\isub\varfive\varsix}\esub\var{\rename{(\la\varfour\tmtwo)}\varthree} & ,\esub\var{\vartwo\varthree}\cons\genv)
      \\[3pt]\hhline{~-|-|-}
          = &(\slicetwo \cons \sliceentry\tmthree\varfive &, \evctxp\varsix\esub\var{\rename{(\la\varfour\tmtwo)}\varthree} & ,\esub\var{\vartwo\varthree}\cons\genv)        
      \\[3pt]\hhline{~-|-|-}
          = &(\slicetwo \cons \sliceentry\tmthree\varfive &, \tm\esub\var{\rename{(\la\varfour\tmtwo)}\varthree} & ,\esub\var{\vartwo\varthree}\cons\genv)        
      \end{array} \]
    \end{itemize}
    
    \item Let $\state = \threestate\slice{\tm\esub\var{(\la\vartwo{\evctxp{\varthree}})\varfour}}\genv \tomachm \threestate{\slice\cons\sliceentry\tm\var}{\evctxp{\varthree}\isub{\vartwo}{\varfour}}\genv = \statetwo$. We proceed by induction on $\slice$. Cases:
    \begin{itemize}
      \item \emph{Empty slice stack}, \ie $\slice = \emslice$. We have: \[
      \begin{array}{rll}
      \decode{\threestate\emslice{\tm\esub\var{(\la\vartwo{\evctxp{\varthree}})\varfour}}\genv} & = & \decodep\genv{\tm\esub\var{(\la\vartwo{\evctxp{\varthree}})\varfour}} 
      \\[3pt]
      & \tom & \decodep\genv{\evctxp{\tm\isub\var\varthree}\isub\vartwo\varfour} 
      \\[3pt]
      & = & \decode{\threestate{\emslice}{\evctxp{\tm\isub\var\varthree}\isub\vartwo\varfour}\genv} 
      \\[3pt]
      & =_{\vartwo \notin \fv\tm} & \decode{\threestate{\emslice\cons\sliceentry\tm\var}{\evctxp{\varthree}\isub{\vartwo}{\varfour}}\genv} 
      \end{array} \]
      Note that $\vartwo\notin\fv\tm$ is ensured by the well-bound invariant (\reflemmap{invariants}{bv}). Moreover, the contextual read-back invariant (\reflemmap{invariants}{rb}) that guarantees that $\decode\genv$ is a right context.
      
      \item \emph{Non-empty slice stack}, \ie $\slice = \slicetwo \cons \sliceentry\tmtwo\varfive$. Let $\tm = \evctxtwop\varsix$. We have: 
      \[\footnotesize\begin{array}{rllr}
      \hhline{~-|-|-}
            &(\slicetwo \cons \sliceentry\tmtwo\varfive &, \tm\esub\var{(\la\vartwo{\evctxp{\varthree}})\varfour} & ,\genv)
      \\[3pt]\hhline{~-|-|-}
        = & 
      (\slicetwo \cons \sliceentry\tmtwo\varfive & ,\evctxtwop\varsix\esub\var{(\la\vartwo{\evctxp{\varthree}})\varfour} & ,\genv)
      \\[3pt]\hhline{~-|-|-}
        = & 
      (\slicetwo & ,\evctxtwop{\tmtwo\isub\varfive\varsix}\esub\var{(\la\vartwo{\evctxp{\varthree}})\varfour}& ,\genv)
      \\[3pt]\hhline{~-|-|-}
        (\text{by \ih}) \tom & (\slicetwo\cons\sliceentry{\evctxtwop{\tmtwo\isub\varfive\varsix}}\var & ,\evctxp{\varthree}\isub{\vartwo}{\varfour} & ,\genv)
      \\[3pt]\hhline{~-|-|-}
       =_{\vartwo \notin \fv{\evctxtwop{\tmtwo\isub\varfive\varsix}}} & (\slicetwo & ,\evctxp{\evctxtwop{\tmtwo\isub\varfive\varsix}\isub\var\varthree}\isub\vartwo\varfour & ,\genv)
      \\[3pt]\hhline{~-|-|-}
        =_{\var \notin \fv\tmtwo} &       
      (\slicetwo \cons \sliceentry\tmtwo\varfive & ,\evctxp{\evctxtwop\varsix\isub\var\varthree}\isub\vartwo\varfour & ,\genv)
      \\[3pt]\hhline{~-|-|-}
        = &       
      (\slicetwo \cons \sliceentry\tmtwo\varfive & ,\evctxp{\tm\isub\var\varthree}\isub\vartwo\varfour & ,\genv)
      \\[3pt]\hhline{~-|-|-}
        =_{\vartwo \notin \fv\tm} &       
	 (\slicetwo \cons \sliceentry\tmtwo\varfive\cons\sliceentry\tm\var & ,\evctxp{\varthree}\isub{\vartwo}{\varfour} & ,\genv)
      \end{array} \]
    \end{itemize}
    Note that the three conditions about variables justifying the equalities are ensured by the well-bound invariant (\reflemmap{invariants}{bv}).
  \end{enumerate}
  
  \item Cases of search transitions:
  \begin{itemize}
    \item $\state = \threestate\slice{\tm\esub\var{\la\vartwo\tmtwo}}\genv \tomachseaone \threestate\slice\tm{\esub\var{\la\vartwo\tmtwo}\cons\genv} = \statetwo$. Cases of $\slice$:
    \begin{itemize}
      \item \emph{Empty slice stack}, \ie $\slice = \emslice$. Then:
      \[\begin{array}{lllll}
       \decode{\threestate\emslice{\tm\esub\var{\la\vartwo\tmtwo}}\genv} 
       & = & \decodep\genv{\tm\esub\var{\la\vartwo\tmtwo}} 
       \\[3pt]
       & = & \decodep{\esub\var{\la\vartwo\tmtwo}\cons\genv}{\tm} 
       & = & \decode{\threestate\emslice\tm{\esub\var{\la\vartwo\tmtwo}\cons\genv}}.
       \end{array}\]
      
      \item \emph{Non-empty slice stack}, \ie $\slice = \slicetwo \cons \sliceentry\tmthree\varthree$. Let $\tm = \evctxp\varfour$. We have 
        \[\begin{array}{rllr}
      \hhline{~-|-|-}
            &(\slicetwo \cons \sliceentry\tmthree\varthree &, \tm\esub\var{\la\vartwo\tmtwo} & ,\genv)
      \\[3pt]\hhline{~-|-|-}
          =  &(\slicetwo \cons \sliceentry\tmthree\varthree &, \evctxp\varfour\esub\var{\la\vartwo\tmtwo} & ,\genv)
      \\[3pt]\hhline{~-|-|-}
          = &(\slicetwo &, \evctxp{\tmthree\isub\varthree\varfour}\esub\var{\la\vartwo\tmtwo} & ,\genv)
      \\[3pt]\hhline{~-|-|-}
          =_{\ih}  &(\slicetwo &, \evctxp{\tmthree\isub\varthree\varfour} & ,\esub\var{\la\vartwo\tmtwo}\cons\genv)
      \\[3pt]\hhline{~-|-|-}
          = &(\slicetwo \cons \sliceentry\tmthree\varthree &, \evctxp\varfour & ,\esub\var{\la\vartwo\tmtwo}\cons\genv)        
      \\[3pt]\hhline{~-|-|-}
          = &(\slicetwo \cons \sliceentry\tmthree\varthree &, \tm & ,\esub\var{\la\vartwo\tmtwo}\cons\genv)        
      \end{array} \]
    \end{itemize}
    
    \item $\state = \threestate\slice{\tm\esub\var{\vartwo\varthree}}\genv \tomachseatwo \threestate\slice\tm{\esub\var{\vartwo\varthree}\cons\genv} = \statetwo$ with either $\vartwo\notin\dom\genv$ or $\genv(\vartwo)\neq\la\varfour\tmtwo$. Cases of $\slice$:
    \begin{itemize}
      \item \emph{Empty slice stack}, \ie $\slice = \emslice$. Then:
      \[ \decode{\threestate\emslice{\tm\esub\var{\vartwo\varthree}}\genv} = \decodep\genv{\tm\esub\var{\vartwo\varthree}} = \decodep{\esub\var{\vartwo\varthree}\cons\genv}{\tm} = \decode{\threestate\emslice\tm{\esub\var{\vartwo\varthree}\cons\genv}}.\]
      
      \item \emph{Non-empty slice stack}, \ie $\slice = \slicetwo \cons \sliceentry\tmtwo\varfour$. Let $\tm = \evctxp\varfive$. We have:
      \[\begin{array}{rllr}
      \hhline{~-|-|-}
            &(\slicetwo \cons \sliceentry\tmtwo\varfour &, \tm\esub\var{\vartwo\varthree} & ,\genv)
      \\[3pt]\hhline{~-|-|-}
          =  &(\slicetwo \cons \sliceentry\tmtwo\varfour &, \evctxp\varfive\esub\var{\vartwo\varthree} & ,\genv)
      \\[3pt]\hhline{~-|-|-}
          = &(\slicetwo &, \evctxp{\tmtwo\isub\varfour\varfive}\esub\var{\vartwo\varthree} & ,\genv)
      \\[3pt]\hhline{~-|-|-}
          =_{\ih}  &(\slicetwo &, \evctxp{\tmtwo\isub\varfour\varfive} & ,\esub\var{\vartwo\varthree}\cons\genv)
      \\[3pt]\hhline{~-|-|-}
          = &(\slicetwo \cons \sliceentry\tmtwo\varfour &, \evctxp\varfive & ,\esub\var{\vartwo\varthree}\cons\genv)        
      \\[3pt]\hhline{~-|-|-}
          = &(\slicetwo \cons \sliceentry\tmtwo\varfour &, \tm & ,\esub\var{\vartwo\varthree}\cons\genv)        
      \end{array} \]
      
    \end{itemize}
    
    \item $\state = \threestate{\slice\cons\sliceentry\tm\var}\varthree\genv \tomachseathree \threestate{\slice}{\tm\isub\var\varthree}\genv = \statetwo$. By the very definition of read-back $\decode{\threestate{\slice\cons\sliceentry\tm\var}{\varthree}\genv} \defeq \decode{\threestate{\slice}{\tm\isub\var\varthree}\genv}$, the statement holds.
  \end{itemize}
  
  \item Any $\leadsto_{\symfont{sea}_1, \symfont{sea}_2, \symfont{sea}_3}$-sequence is bound by the sum of the size of the active code plus the sizes of the terms in the slice stack because:
  \begin{itemize}
  \item  $\leadsto_{\symfont{sea}_1}$ and $\leadsto_{\symfont{sea}_2}$ consume \ESs from the active code, 
  \item $\leadsto_{\symfont{sea}_3}$ consumes an entry from the slice when there are no ESs left on the active code.
  \end{itemize}
  
  \item Let $\state = \threestate\slice\tm\genv$ be a final state. First, the active code $\tm$ has be a variable $\var$. Otherwise, one among transitions $\symfont{sea}_1, \symfont{sea_2}, \esym, \msym$ would apply. Therefore, the slice stack $\slice$ has to be empty, otherwise transition $\symfont{sea}_3$ would apply.
  
  By the contextual read-back invariant (\reflemmap{invariants}{rb}), $\decode\genv$ is a right context, which implies that $\decode\state = \decodep\genv\var$ is $\topos$-normal.\qed
\end{enumerate}
\end{proof}